\documentclass[reqno]{amsart}

\usepackage{amsmath}
\usepackage{amssymb,latexsym}
\usepackage[mathscr]{eucal}
\usepackage{dsfont}
\usepackage{url}
\usepackage{epsfig,graphicx,subfigure}

\RequirePackage[colorlinks,citecolor=blue,urlcolor=blue]{hyperref}

\numberwithin{equation}{section}

\title[Partial Distance Correlation]{Partial Distance Correlation with
Methods for Dissimilarities}

\author{G\'abor J. Sz\'ekely}
\address{G\'abor J. Sz\'ekely\\
National Science Foundation\\
4201 Wilson Blvd. \#1025\\
  Arlington, VA 22230}
  \email{gszekely@nsf.gov}
\thanks{G\'abor J. Sz\'ekely is Program Director in Statistics at the National
Science Foundation, 4201 Wilson Blvd. \#1025, Arlington, VA 22230,
email: gszekely@nsf.gov.}
\thanks{Maria L. Rizzo is Associate Professor
in the Dept. of Mathematics and Statistics at
Bowling Green State University, Bowling Green, OH 43403,
email: mrizzo@bgsu.edu.}
\thanks{
{Research of the first author was supported by the National Science Foundation,
while working at the Foundation.}}

\author{Maria L. Rizzo}
\address{Maria L. Rizzo\\
Department of Mathematics and Statistics\\
Bowling Green State University\\
Bowling Green, OH 43403}
\email{mrizzo@bgsu.edu}
\urladdr{personal.bgsu.edu/~mrizzo}

    \numberwithin{equation}{section}

    \theoremstyle{plain}
    
    \newtheorem{thm}{Theorem}
    \newtheorem{Lemma}{Lemma}
    \newtheorem{Definition}{Definition}
    
    \newtheorem{Proposition}{Proposition}

    \theoremstyle{remark}
    \newtheorem{remark}{remark}%
    \newtheorem{Example}{Example}%

    \newcommand{\la}{\langle}
    \newcommand{\ra}{\rangle}

    \newcommand{\fx}{\phi_{{}_{X}}}
    \newcommand{\fy}{\phi_{{}_{Y}}}
    \newcommand{\fxy}{\phi_{{}_{X,Y}}}

    \def\Cov{\mathop{\rm Cov}\nolimits}%

    \def\dCov{\mathop{\rm dCov}\nolimits}%

    \def\pdCov{\mathop{\rm pdCov}\nolimits}%
    \def\pdCor{\mathop{\rm pdCor}\nolimits}%

\date {\today}

\begin{document}

\maketitle

\begin{abstract}
Distance covariance and distance correlation are scalar coefficients that characterize independence of random
vectors in arbitrary dimension. Properties, extensions, and applications of distance correlation have been
discussed in the recent literature, but the problem of defining the partial distance correlation has remained an
open question of considerable interest. The problem of partial distance correlation is more complex than partial
correlation partly because the squared distance covariance is not an inner product in the usual linear space. For
the definition of partial distance correlation we introduce a new Hilbert space where the squared distance
covariance is the inner product.  We define the partial distance correlation statistics with the help of this
Hilbert space, and develop and implement a test for zero partial distance correlation. Our intermediate results
provide an unbiased estimator of squared distance covariance, and a neat solution to the problem of distance
correlation for dissimilarities rather than distances.
\end{abstract}

\subjclass[2000]
{Primary: 62H20, 62H15; Secondary: 62Gxx}

\keywords{independence, multivariate, partial distance correlation, dissimilarity, energy statistics}

\section{Introduction}\label{S1}

Distance covariance (dCov) and distance correlation characterize multivariate independence for random vectors in
arbitrary, not necessarily equal dimension. In this work we focus on the open problem of \emph{partial distance
correlation}. Our intermediate results include methods for applying distance correlation to dissimilarity
matrices.

The distance covariance, denoted $\mathcal V(X, Y)$, of two random vectors $X$ and $Y$ characterizes
independence; that is, $$
 \mathcal V(X, Y) \geq 0
$$ with equality to zero if and only if $X$ and $Y$ are independent. This coefficient is defined by a weighted $L_2$ norm measuring the distance between
the joint characteristic function (c.f.) $\fxy$ of $X$ and $Y$, and the
product $\fx \fy$ of the marginal c.f.'s of $X$ and $Y$.
If $X$ and $Y$ take values in $\mathbb R^p$ and $\mathbb R^q$, respectively,
$\mathcal V(X, Y)$ is the non-negative square root of
 \begin{align}\label{e:Aw}
\mathcal{V}^2(X,Y)&= \| \fxy(t,s)-\fx(t)\fy(s) \|_w^2 \\&:=
\int_{{R}^{p + q}}|\fxy (t,s) - \fx(t)\fy(s)|^2 w (t, s) \,
dt \, ds, \notag
\end{align}
where
$w(t, s) := (|t|_p^{1+p} |s|_q^{1+ q})^{-1}$.
The integral exists
provided that $X$ and $Y$ have finite first moments.
Note that Feuerverger (1993) \cite{feuer1993} proposed a bivariate test based on this idea and applied the same weight function $w$, where it may have first appeared.

The following identity is established in Sz\'ekely and Rizzo \cite[Theorem 8, p. 1250]{sr09a}. Let $(X, Y)$, $(X', Y')$, and $(X'', Y'')$ be independent and identically distributed (iid), each with joint distribution $(X, Y)$. Then
\begin{align} \label{Edcov}
\mathcal V^2(X, Y) &= E|X-X'||Y-Y'| + E|X-X'| \cdot E|Y-Y'| \\
& \qquad - E|X-X'||Y-Y''| - E|X-X''||Y-Y'|, \notag
\end{align}
provided that $X$ and $Y$ have finite first moments.
In Section \ref{S4} an alternate version of (\ref{Edcov}) is defined for $X$ and $Y$ taking values in a separable Hilbert space. That definition and intermediate results lead to the definition of partial distance covariance.

We summarize a few key properties and computing formulas below for easy reference. The distance correlation
(dCor) $\mathcal R(X, Y)$ is a standardized coefficient, $0 \leq \mathcal R(X, Y) \leq 1$,
 that also characterizes independence:
$$
 \mathcal R(X, Y) =
\left\{
  \begin{array}{ll}
  \frac{\mathcal V(X, Y)}{\sqrt{\mathcal V(X,X) \mathcal V(Y,Y)}}, &
  \hbox{$\mathcal V(X,X) \mathcal V(Y,Y) > 0$;} \\
  0, & \hbox{$\mathcal V(X,X) \mathcal V(Y,Y) = 0$.}
  \end{array}
\right.
$$
For more details see \cite{srb07} and \cite{sr09a}. Properties, extensions,
and applications of distance correlation have been discussed in the recent literature; see e.g. Dueck et al.\ \cite{degr2012}, Lyons
\cite{lyons2013}, Kong et al.\ \cite{kkklw2012}, and Li, Zhong, and Zhu \cite{lzz2012}. However, there is
considerable interest among researchers and statisticians on the open problem of a suitable definition and
supporting theory of partial distance  correlation. Among the many potential application areas of partial
distance correlation are variable selection (see Example \ref{ex7})
and graphical models; see e.g.\ Wermuth and Cox \cite{wc13} for an example of work that motivated the question in that context.

In this work we introduce the definition of partial distance covariance
(pdCov) and partial distance correlation
(pdCor) statistics and population coefficients. First, let us see why it is
not straightforward to generalize
distance correlation to partial distance correlation in a meaningful way
that preserves the essential properties
one would require, and allows for interpretation and inference.

One could try
to follow the definitions of the
classical partial covariance and partial correlation that are based on
orthogonal projections in a Euclidean
space, but there is a serious difficulty. Orthogonality in case of partial
distance covariance and partial
distance correlation means independence, but when we compute the orthogonal
projection of a random variable onto
the condition variable, the ``remainder'' in the difference is typically not
independent of the condition.

Alternately, the form of sample distance covariance (Definition \ref{defVn})
may suggest an inner product, so one might think of
working in the Hilbert space of double centered distance matrices
(\ref{Dcenter}), where the inner
product is the squared distance covariance statistic (\ref{e:AnXY}).
Here we are facing another problem: what would the projections
represent? The difference $D$ of double centered distance matrices
is typically not a double centered distance
matrix of any sample. This does not affect formal computations,
but if we cannot interpret our formulas
in terms of samples then inference becomes impossible.

To overcome these
difficulties while preserving the essential properties of distance
covariance, we finally arrived at an elegant
solution which starts with defining an alternate type of double centering
called ``$\mathcal U$-centering'' (see
Definition \ref{def:tildeA} and Proposition \ref{P:unbiased} below). The corresponding inner product is an
unbiased estimator of squared population distance covariance. In the Hilbert
space of ``$\mathcal U$-centered'' matrices,
all linear combinations, and in particular projections, are zero diagonal $\mathcal U$-centered matrices. We prove a representation theorem that connects
the orthogonal projections to random samples in Euclidean space.
Methods for inference are
outlined and implemented, including methods for non-Euclidean
dissimilarities.

Definitions and background for dCov and dCor are summarized in Section 2.
The partial distance correlation
statistic is introduced in Section 3, and Section 4 covers the population
partial distance covariance and
inference for a test of the hypothesis of zero partial distance
correlation. Examples and Applications are given
in Section 5, followed by a Summary in Section 6. Appendix A contains
proofs of statements.

\section{Background}\label{S2}

In this section we summarize the definitions for distance covariance and
distance correlation statistics.

\subsection{Notation}
Random vectors are denoted with upper case letters and their sample
realizations with lowercase letters. In
certain contexts we need to work with a data matrix, which is denoted in
boldface, such as $\mathbf U$ or
$\mathbf V$. For example, the sample distance correlation of a realization
from the joint distribution of $(X,
Y)$ is denoted $\mathcal R_n(\mathbf X, \mathbf Y)$. The norm $|\cdot|_d$
denotes the Euclidean norm when its
argument is in $\mathbb R^d$, and we omit the subscript if the meaning
is clear in context. A primed random
variable denotes an independent and identically distributed (iid) copy of the unprimed symbol; so $X,X'$ are iid,
etc. The pair $x, x'$ denotes two realizations of the random variable $X$. The transpose of $x$ is denoted by
$x^T$.
 Other notation will be
introduced within the section where it is needed.

In this paper a \emph{random sample} refers to independent and identically distributed (iid) realizations from
the underlying joint or marginal distribution.

\subsection{Sample dCov and dCor}

The distance covariance and distance correlation statistics are functions of the double centered distance
matrices of the samples. For an observed random sample $\{(x_{i}, y_{i}):i=1,\dots,n\} $ from the joint
distribution of random vectors $X $ and $Y $, compute the Euclidean distance matrices $(a_{ij})=(|x_i-x_j|_p)$
and $(b_{ij})=(|y_i-y_j|_q)$.  Define
\begin{align}\label{Dcenter}
\widehat A_{ij}= a_{ij}-\bar a_{i.}- \bar a_{.\,j} + \bar a_{..}\,,
\qquad i,j=1,\dots,n,
\end{align}
where
\begin{align*}
\bar a_{i .}= \frac{1}{n}
\sum_{j=1}^n a_{ij}, \quad \bar a_{. j}, = \frac{1}{n}
\sum_{i=1}^n
a_{ij}, \quad
\bar a_{..} & = \frac{1}{n^2} \sum_{i,j=1}^n a_{ij}.
\end{align*}
Similarly, define $\widehat B_{ij}= b_{ij} -\bar b_{i .}- \bar b_{. \,j} +  \bar b_{..}$, for $i,j=1,\dots,n$.

\begin{remark}
In previous work, we denoted the double centered distance matrix of a sample by a plain capital letter, so that
$\widehat A$ and $\widehat B$ above were denoted by $A$ and $B$. In this paper we work with two different methods
of centering, so we have introduced notation to clearly distinguish the two methods. In the first centering above
the centered versions $\widehat A_{ij}$ have the property that all rows and columns sum to zero. Another type of
centering, which we will call unbiased or $\mathcal U$-centering and denote by $\widetilde A_{ij}$ in Definition
\ref{def:tildeA} below, has the additional property that all expectations are zero: that is, $E[\widetilde
A_{ij}] = 0$ for all $i,j$. Although this centering has a somewhat more complicated formula (\ref{AU}), we will
see in this paper the advantages of the new centering.
\end{remark}

\begin{Definition}\label{defVn}
The sample distance covariance $\mathcal{V}_n(\mathbf X, \mathbf Y)$ and sample distance correlation $\mathcal
R_n(\mathbf X, \mathbf Y)$ are defined by
\begin{equation}\label{e:AnXY}
 \mathcal{V}^2_n (\mathbf X, \mathbf Y) = {\frac{1}{n^2} \sum_{i,j=1}^n  \widehat A_{ij} \widehat B_{ij}} \;.
\end{equation}
and
\begin{equation}\label{e:alphan}
  \mathcal{R}^2_n({\mathbf X, \mathbf Y})=
 \left\{
   \begin{array}{ll}
 \frac {\mathcal{V}^2_n({\mathbf X, \mathbf Y})}{\sqrt {\vphantom{\widehat{A}}
 \mathcal{V}^2_n (\mathbf X) \mathcal{V}^2_n(\mathbf Y)}} \;,
   & \hbox{$ \mathcal{V}^2_n(\mathbf X)  \mathcal{V}^2_n (\mathbf Y) >0
$;} \\
     0, & \hbox{$\mathcal{V}^2_n (\mathbf X)  \mathcal{V}^2_n (\mathbf Y) =0$.}
   \end{array}
 \right.
 \end{equation}
respectively, where the squared sample distance variance is defined by
\begin{equation}\label{e:AnX}
 \mathcal{V}^2_n (\mathbf X) = \mathcal{V}^2_n ({\mathbf X, \mathbf X}) = {\frac{1}{n^2} \sum_{i,j=1}^n
 \widehat A_{ij}^2} \;.
 \end{equation}
\end{Definition}

\begin{remark}
Although it is obvious, it is perhaps worth noting that $\mathcal V_n(\mathbf X, \mathbf Y)$ and $\mathcal
R_n(\mathbf X, \mathbf Y)$ are rigid motion invariant, and they can be defined entirely in terms of the distance
matrices, or equivalently in terms of the double centered distance matrices. This fact is important for what
follows, because as we will see, one can compute partial distance covariance by operating on certain
transformations of the distance matrices, without reference to the original data that generated the matrices. (These definitions are not invariant to monotone
transformation of coordinates, such as ranks.)
\end{remark}

If $X$ and $Y$ have finite first moments, the population distance covariance coefficient $\mathcal V^2(X, Y)$
exists and equals zero if and only if the random vectors $X$ and $Y$
are independent.
Some of the properties of distance covariance and distance correlation include:
\begin{enumerate}
 \item $\mathcal V_n(\mathbf X, \mathbf Y)$ and $\mathcal R_n(\mathbf X, \mathbf Y)$ converge almost surely
to $\mathcal V(X, Y)$ and $\mathcal R(X, Y)$, as $n \to \infty$.
 \item $\mathcal V_n(\mathbf X, \mathbf Y) \geq 0$ and $\mathcal V_n(\mathbf X)=0$ if and only if every
     sample observation is identical.
\item \label{Th4-1}
 $0 \leq \mathcal{R}_n(\mathbf X, \mathbf Y) \leq 1$.
\item \label{Th4-2} If $\mathcal{R}_n(\mathbf X, \mathbf Y) = 1$ then
 there exists a vector $a$, a non-zero real number b and an orthogonal matrix
 $R$ such that $\mathbf Y = a + b\mathbf XR$, for the data matrices $\mathbf X$ and $\mathbf Y$.
\end{enumerate}

A consistent test of multivariate independence is based on the sample distance covariance: large values of $n
\mathcal V^2_n(\mathbf X, \mathbf Y)$ support the alternative hypothesis that
$X$ and $Y$ are dependent (see
\cite{srb07,sr09a}). A high-dimensional distance correlation
 $t$-test of independence was introduced by Sz\'ekely and Rizzo \cite{sr2013a};
the tests are implemented as \texttt{dcov.test} and \texttt{dcor.ttest}, respectively, in the \emph{energy} package
(Rizzo and Sz\'ekely, \cite{energy}) for R (R Core Team \cite{R}).

The definitions of sample distance covariance and sample distance correlation can be extended to random samples
taking values in any, possibly different, metric spaces. For defining the population coefficient we need to
suppose more (see Lyons \cite{lyons2013}).  If $X$ and $Y$ take values in possibly different separable Hilbert
spaces and both $X$ and $Y$ have finite expectations then it remains true that  $\mathcal V^2(X, Y) \ge 0 $, and
equals zero if and only if $X$ and $Y$ are independent. This implies that the results of this paper can be
extended to separable Hilbert space valued variables. Extensions and population coefficients pdCov and pdCor are
discussed in Section \ref{S4}. Extending Lyons \cite{lyons2013},
Sejdinovic et al. \cite{ssgf2013} discuss the equivalence of distance
based and RKHS-based statistics for testing dependence.

For theory, background, and further properties of the population and sample coefficients, see Sz\'ekely, et al. \cite{srb07}, Sz\'ekely and Rizzo \cite{sr09a}, and Lyons \cite{lyons2013}, and the references therein.
On the weight function applied see also
Feuerverger \cite{feuer1993} and Sz\'ekely and Rizzo \cite{sr2012}.
For an overview of recent methods for measuring
dependence and testing independence of random vectors, including distance covariance, readers are referred to
Josse and Holmes \cite{jh2013}.

\section{Partial distance correlation}\label{S3}

We introduce partial distance correlation by first considering the sample coefficient.

\subsection{The Hilbert space of centered distance matrices}

\begin{Definition}\label{def:tildeA}
Let $A=(a_{ij})$ be a symmetric, real valued $n \times n$ matrix with zero diagonal, $n > 2$. Define the
$\mathcal U$-centered matrix $\widetilde A$ as follows. Let the $(i, j)$-th entry of $\widetilde A$ be defined by
\begin{equation} \label{AU}
\widetilde A_{i,j}=
\left\{
     \begin{array}{ll}
       a_{i,j} - \frac 1{n-2} \sum\limits_{\ell=1}^n a_{i, \ell} -
\frac{1}{n-2} \sum\limits_{k=1}^n a_{k, j} + \frac{1}{(n-1)(n-2)} \sum\limits_{k, \ell=1}^n a_{k, \ell},
 & \hbox{$i \neq j$;} \\
       0, & \hbox{$i=j$.}
     \end{array}
   \right.
\end{equation}
\end{Definition}
Here ``$\mathcal U$-centered'' is so named because as shown below, the corresponding inner product (\ref{VU}) defines an
unbiased estimator of squared distance covariance.

\begin{Proposition}\label{P:unbiased}
Let $(x_i,y_i), \, i=1,\dots,n$ denote a sample of observations from the
joint distribution $(X,Y)$ of random vectors $X$ and $Y$.
Let $A=(a_{ij})$ be the Euclidean distance matrix of the sample $x_1,\dots,$ $x_n$ from the distribution of $X$,
and $B=(b_{ij})$ be the Euclidean distance matrix of the sample $y_1,\dots, y_n$ from the distribution of $Y$.
Then if $E(|X|+|Y|)<\infty$, for $n > 3$,
\begin{equation} \label{VUdcov}
(\widetilde A \cdot \widetilde B) :=  \frac{1}{n(n-3)} \sum_{i \neq j} \widetilde A_{i,j}\widetilde B_{i,j}
\end{equation}
is an unbiased estimator of squared population distance covariance $\mathcal V^2(X, Y)$.
\end{Proposition}
Proposition \ref{P:unbiased} is proved in Appendix (\ref{prfUnbiased}). It is obvious that $\widetilde A=0$ if all of the sample observations are identical.
More generally, $\widetilde A =0$ if and only if the $n$ sample observations are equally distant or at least $n-1$ of the $n$ sample observations are identical.

For a fixed $n \geq 4$, we define a Hilbert space generated by
Euclidean distance matrices of arbitrary sets (samples) of $n$ points in a
Euclidean space $\mathbb R^p$, $p \geq 1$. Consider the linear span
$\mathscr S_{n}$ of all $n \times n$ distance matrices of samples
$\{x_1,\dots,x_n\}$. Let $A = (a_{ij})$ be an
arbitrary element in $\mathscr S_{n}$. Then $A$ is a real valued,
symmetric matrix with zero diagonal.

Let $\mathcal H_{n}=\{\widetilde A: A \in \mathscr S_{n} \}$ and for each pair
of elements $C=(C_{i,j})$,
$D=(D_{i,j})$ in the linear span of $\mathcal H_{n}$ define their inner product
\begin{equation}\label{VU}
  (C \cdot D) =
  \frac{1}{n(n-3)} \sum_{i \neq j} C_{ij} D_{ij}.
\end{equation}
If $(C \cdot C)=0$ then $(C \cdot D)=0$ for any $D \in \mathcal H_{n}$.

Below in Theorem \ref{thmMDS} it is shown that every matrix $C \in \mathscr H_n$
is the $\mathcal U$-centered distance
matrix of a configuration of $n$ points in a Euclidean space $\mathbb R^p$,
where $p \leq n-2$.

\begin{thm}\label{T:Hilbert}
The linear span of all $n \times n$ matrices $\mathcal H_n= \{\widetilde A: A \in \mathscr S_{n}\}$ is a Hilbert
space with inner product defined by (\ref{VU}).
\end{thm}

\begin{proof} Let $\mathscr H_n$ denote the linear span of $\mathcal H_n$.
If $\widetilde A, \widetilde B \in \mathcal H_n$ and $c, d \in \mathbb R$, then $(c\widetilde A + d\widetilde
B)_{i,j} =\widetilde{ (cA+dB)}_{i,j},$ so $(c\widetilde A + d\widetilde B) \in \mathscr H_n$. It is also true
that $C \in \mathscr H_n$ implies that $-C \in \mathscr H_n$, and the zero element is the $n \times n$ zero
matrix.

Then since $(c\widetilde A + d\widetilde B)=\widetilde{ (cA+dB)},$ for the inner product, we only need to prove
that for $\widetilde A, \widetilde B, \widetilde C \in \mathcal H_n$ and real constants $c$, the following
statements hold:
\begin{enumerate}
  \item $(\widetilde A \cdot \widetilde A) \geq 0$.
  \item $(\widetilde A \cdot \widetilde A)=0$ only if $\widetilde A = 0$.
  \item $(\widetilde{(cA)} \cdot \widetilde B) = c (\widetilde A \cdot \widetilde B)$
  \item $((\widetilde A + \widetilde B) \cdot \widetilde C) = (\widetilde A \cdot \widetilde C) + (\widetilde
      B \cdot \widetilde C)$.
\end{enumerate}
Statements (i) and (ii) hold because $(\widetilde A \cdot \widetilde A)$ is proportional to a sum of squares.
Statements (iii) and (iv) follow easily from the definition of $\mathscr H_n$, $\widetilde A$, and (\ref{VU}).
\end{proof}

The space $\mathscr H_n$ is finite dimensional because it
is a subspace of the space of all symmetric, zero-diagonal
$n \times n$ matrices.

In what follows, $\mathscr H_n$ denotes the Hilbert space of Theorem \ref{T:Hilbert} with inner product
(\ref{VU}), and $| \widetilde A | = (\widetilde A, \widetilde A)^{1/2}$ is the norm of $\widetilde A$.

\subsection{Sample pdCov and pdCor}

From Theorem \ref{T:Hilbert} a projection operator (\ref{proj})
can be defined in the Hilbert space
$\mathscr H_n$, $n \geq 4$, and applied to define partial distance covariance and partial distance correlation
for random vectors in Euclidean spaces.
Let $\widetilde A, \widetilde B,$ and $\widetilde C$ be elements of
$\mathscr H_n$ corresponding to samples $x, y,$ and $z$, respectively, and let
\begin{equation}\label{proj}
  P_{z^\perp}(x) = \widetilde A - \frac{(\widetilde A \cdot \widetilde C)}{(\widetilde C \cdot \widetilde C)} \,
  \widetilde C,
\qquad
  P_{z^\perp}(y) = \widetilde B - \frac{(\widetilde B \cdot \widetilde C)}{(\widetilde C \cdot \widetilde C)} \,
  \widetilde C,
\end{equation}
denote the orthogonal projection of $\widetilde A(x)$ onto $(\widetilde C(z))^\perp$ and the orthogonal
projection of $\widetilde B(y)$ onto $(\widetilde C(z))^\perp$, respectively. In case $(\widetilde C \cdot
\widetilde C)=0$ the projections are defined $P_{Z^\perp}(x) = \widetilde A$ and $P_{Z^\perp}(y) = \widetilde B$.
Clearly $P_{z^\perp}(x)$ and $P_{z^\perp}(y)$ are elements of $\mathscr H_n$, their dot product is defined by
(\ref{VU}), and we can define an estimator of pdCov($X, Y;Z$) via projections.

\begin{Definition}[Partial distance covariance]\label{def.pdcov}
Let $(x,y,z)$ be a random sample observed from the joint distribution of $(X, Y, Z)$.  The sample partial
distance covariance (pdCov) is defined by
\begin{equation} \label{PDCOV}
\pdCov(x, y; z) = (P_{z^\perp}(x) \cdot P_{z^\perp}(y)),
\end{equation}
where $P_{z^\perp}(x)$, and $P_{z^\perp}(y)$ are defined by (\ref{proj}), and
\begin{align}\label{pdcov}
   (P_{z^\perp}(x) \cdot P_{z^\perp}(y)) &= \frac{1}{n(n-3)} \sum_{i \neq j} (P_{z^\perp}(x))_{i,j}
   (P_{z^\perp}(y))_{i,j}.
\end{align}
\end{Definition}

\begin{Definition}[Partial distance correlation]
Let $(x,y,z)$ be a random sample observed from the joint distribution of $(X, Y, Z)$. Then sample partial
distance correlation is defined as the cosine of the angle $\theta$ between the `vectors' $P_{z^\perp}(x)$
and $P_{z^\perp}(y)$ in the Hilbert space $\mathscr H_n$:
\begin{align}\label{pdcor2}
R^*(x, y;z) := \cos \theta &=
  \frac{(P_{z^\perp}(x) \cdot P_{z^\perp}(y))}{|P_{z^\perp}(x)||P_{z^\perp}(y)|},
  \qquad |P_{z^\perp}(x)||P_{z^\perp}(y)| \neq 0,
\end{align}
and otherwise $R^*(x, y;z):=0$.
\end{Definition}

\subsection{Representation in Euclidean space}

Sample pdCov and pdCor have been defined via projections in the Hilbert space generated by $\mathcal U$-centered distance matrices. In this section, we start to
consider the interpretation of sample pdCov and sample pdCor.

Since pdCov is defined as the
inner product (\ref{pdcov}) of two $\mathcal U$-centered matrices, and (unbiased squared) distance covariance (\ref{VUdcov}) is computed as inner product, a natural question is the following. Are matrices $P_{z^\perp}(x)$ and $P_{z^\perp}(y)$ the $\mathcal U$-centered \emph{Euclidean} distance matrices of samples of points ${\mathbf U}$ and ${\mathbf V}$, respectively? If so, then the sample partial distance covariance (\ref{pdcov}) is distance covariance of $\mathbf U$ and $\mathbf V$, as defined by (\ref{VUdcov}).

For every sample of points ${\mathbf X}=[x_1,\dots,x_n]$, $x_i \in \mathbb R^p$, there is
a $\mathcal U$-centered matrix $\tilde A=\tilde A(x)$ in $\mathscr H_n$. Conversely,
given an arbitrary element $H$ of $\mathscr H_n$, does there exist a configuration
of points ${\mathbf U}=[u_1,\dots,u_n]$ in some Euclidean space $\mathbb R^q$, for some $q \geq 1$, such that the $\mathcal U$-centered Euclidean distance matrix of sample
${\mathbf U}$ is exactly equal to the matrix $H$?
In this section we prove that the answer is yes: $P_{z^\perp}(x)$, $P_{z^\perp}(y)$ of (\ref{pdcov}), and in general every element in $\mathscr H_n$, is the $\mathcal U$-centered distance matrix of some sample of $n$ points in a Euclidean space.

First a few properties of centered
distance matrices are established in the following lemma.

\begin{Lemma}\label{lemmaU}
Let $\widetilde A$ be a $\mathcal U$-centered distance matrix. Then
\begin{enumerate}
\item Rows and columns of $\widetilde A$ sum to zero.
\item $\widetilde {({\widetilde A})} = \widetilde A$. That is, if $B$ is the matrix obtained by
$\mathcal U$-centering an element $\widetilde A \in \mathscr H_n$, $B = \widetilde A$.
\item $\widetilde A$ is invariant to double centering. That is, if $B$ is the matrix obtained by double
    centering
the matrix $\widetilde A$, then $B=\widetilde A$.
\item \label{L4} If $c$ is a constant and $B$ denotes the matrix
obtained by adding $c$ to the off-diagonal elements of $\widetilde A$, then $\widetilde B = \widetilde A$.
\end{enumerate}
\end{Lemma}
See Appendix \ref{prfLemmaU} for proof of Lemma \ref{lemmaU}.

As Lemma \ref{lemmaU}(iv) is essential for our results, it
becomes clear that we cannot apply double centering as in
the original (biased) definition of distance covariance here.
The invariance with respect to the constant $c$
in (iv) holds for $\mathcal U$-centered matrices
but it does not hold for double centered matrices.

An $n \times n$ matrix $D$ is called \emph{Euclidean} if there exist
points $v_1,\dots,v_n$ in a Euclidean space
such that their Euclidean distance matrix is exactly $D$; that is,
$d_{ij}^2=|v_i-v_j|^2=(v_i-v_j)^T(v_i-v_j)$,
$i,j=1,\dots,n$. Necessary and sufficient conditions that $D$ is
Euclidean are well known results of classical
(metric) MDS. With the help of Lemma \ref{lemmaU}, and certain results
from the theory of MDS, we can find for
each element $H \in \mathscr H_n$ a configuration of points $v_1,\dots,v_n$
in Euclidean space such that their
$\mathcal U$-centered distance matrix is exactly equal to $H$.

A solution to the underlying system of equations to solve for the
points $v$ is found in Schoenberg \cite{schoenberg} and Young and
Householder \cite{yh1938}. It is a classical result that is well known in
(metric) multidimensional scaling. Mardia, Kent, and Bibby \cite{mkb1979}
summarize the result in Theorem 14.2.1, and provide a proof. For an
overview of the methodology see also Cox and Cox \cite{cc2001}, Gower
\cite{gower1966}, Mardia \cite{mardia1978}, and Torgerson \cite{torg1958}.
We apply the converse (b) of Theorem 14.2.1 as stated in Mardia, Kent,
and Bibby \cite{mkb1979}, summarized below.

Let $(d_{ij})$ be a dissimilarity matrix and define $a_{ij}=-\frac{1}{2} d_{ij}^2$. Form the double centered
matrix $\widehat A = (a_{ij}-\bar a_{i.} - \bar a_{.j} + \bar a_{..})$. If $\widehat A$ is positive semi-definite
(p.s.d.) of rank $p$, then a configuration of points corresponding to $\widehat D$ can be constructed as follows.
Let $\lambda_1 \geq \lambda_2 \geq \dots \geq \lambda_p$ be the positive eigenvalues of $\widehat A$, with
corresponding normalized eigenvectors $v_1,\dots,v_p$, such that $v_k^Tv_k=\lambda_k$, $k=1,\dots,p$. Then if $V$
is the $n \times p$ matrix of eigenvectors, the rows of $V$ are $p$-dimensional vectors that have interpoint
distances equal to $(d_{ij})$, and $\widehat A = VV^T$ is the inner product matrix of this set of points. The
solution is constrained such that the centroid of the points is the origin. There is at least one zero eigenvalue
so $p \leq n-1$.

When the matrix $\widehat A$ is not positive semi-definite, this leads us to the \emph{additive constant
problem}, which refers to the problem of finding a constant $c$ such that by adding the constant to all
off-diagonal entries of $(d_{ij})$ to obtain a dissimilarity matrix $D_c$, the resulting double centered matrix
is p.s.d. Let $\widehat A_c(d^2_{ij})$ denote the double centered matrix obtained by double centering $- \frac
12(d^2_{ij}+c(1-\delta^{ij}))$, where $\delta^{ij}$ is the Kronecker delta. Let $\widehat A_c(d_{ij})$ denote the
matrix obtained by double centering $- \frac 12(d_{ij}+c(1-\delta^{ij}))$. The smallest value of $c$ that makes
$\widehat A_c(d^2_{ij})$ p.s.d. is $c^*=-2 \lambda_n$, where $\lambda_n$ is the smallest eigenvalue of $\widehat
A_0(d^2_{ij})$. Then $\widehat A_c(d^2_{ij})$ is p.s.d. for every $c \geq c^*$.
The number of positive eigenvalues of the p.s.d. double centered matrix $\widehat A_c(d^2_{ij})$ determines the
dimension required for the representation in Euclidean space.

However, we require a constant to be added to the elements $d_{ij}$ rather than $d^2_{ij}$. That is, we require a
constant $c$, such that the dissimilarities $d_{ij}^{(c)}=d_{ij}+c(1-\delta^{ij})$ are Euclidean. The solution by
Cailliez \cite{cailliez} is $c^*$, where $c^*$ is the largest eigenvalue of a $2n \times 2n$ block matrix $$
  \begin{bmatrix}
    0 & \widehat A_0(d^2_{ij}) \\
    I & \widehat A_0(d_{ij}) \\
  \end{bmatrix},
$$ where $0$ is the zero matrix and $I$ is the identity matrix of size $n$
 (see Cailliez \cite{cailliez} or
Cox and Cox \cite[Sec.~2.2.8]{cc2001} for details). This result guarantees that
there exists a constant $c^*$ such that the adjusted dissimilarities $d_{ij}^{(c)}$ are Euclidean. In this case at most $n-2$ dimensions are required (Cailliez \cite[Theorem~1]{cailliez}).

Finally, given an arbitrary element $H$ of $\mathscr H_n$, the problem is to find a configuration of points
$\mathbf V=[v_1,\dots,v_n]$ such that the $\mathcal U$-centered distance matrix of $\mathbf V$ is exactly equal to the
element $H$. Thus, if $H=(h_{ij})$ we are able to find points $\mathbf V$ such that the Euclidean distance matrix
of $\mathbf V$ equals $H_c=(h_{ij}+c(1-\delta^{ij}))$, and we need $\widetilde{ H}_c=H$. Now since  $(c\widetilde
A + d\widetilde B)=\widetilde{ (cA+dB)},$ we can apply Lemma \ref{lemmaU} to $H$, and $\widetilde{H}_c=H$ follows
from Lemma \ref{lemmaU}(ii) and Lemma \ref{lemmaU}(iv). Hence, by applying classical MDS with the additive
constant theorem, and Lemma \ref{lemmaU} (ii) and (iv), we obtain the configuration of points $\mathbf V$ such
that their $\mathcal U$-centered distances are exactly equal to the element $H \in \mathscr H_n$. Lemma \ref{lemmaU}(iv)
also shows that the inner product is invariant to the constant $c$.

This establishes our theorem on representation in Euclidean space.

\begin{thm}\label{thmMDS}
Let $H$ be an arbitrary element of the Hilbert space $\mathscr H_n$ of $\mathcal U$-centered distance matrices. Then there
exists a sample $v_1,\dots,v_n$ in a Euclidean space of dimension at most $n-2$, such that the $\mathcal U$-centered
distance matrix of $v_1,\dots,v_n$ is exactly equal to $H$.
\end{thm}

\begin{remark}The above details also serve to illustrate why a Hilbert space
of double centered matrices (as applied in the original, biased statistic $\mathcal V_n^2$) is not applicable for
a meaningful definition of partial distance covariance. The diagonals of double centered distance matrices are
not zero, so we cannot get an exact solution via MDS, and the inner product would depend on $c$. Another problem
is that while $\mathcal V_n^2$ is always non-negative, the inner product of projections could easily be
negative.
\end{remark}

\paragraph{Methods for Dissimilarities} In community ecology and other fields of application,
it is often the case that only the (non-Euclidean) dissimilarity matrices are
available (see e.g.\ the genetic distances of Table \ref{T:maize7}).
Suppose that the dissimilarity matrices are symmetric with zero
diagonal. An application of Theorem \ref{thmMDS} provides methods for
this class of non-Euclidean dissimilarities.
In this case, Theorem \ref{thmMDS} provides that there exists
samples in Euclidean space such that their $\mathcal U$-centered Euclidean distance
matrices are equal to the dissimilarity matrices. Thus, to apply distance
correlation methods to this type of problem, one only needs to obtain the
Euclidean representation. Existing software implementations of classical
MDS can be applied to obtain the representation derived above. For
example, classical MDS based on the method outlined in Mardia
\cite{mardia1978} is implemented in the R function \texttt{cmdscale},
which is in the \emph{stats} package for R. The
\texttt{cmdscale} function includes options to apply the additive constant
of Cailliez \cite{cailliez} and to
specify the dimension. The matrix of points $\mathbf V$ is returned in
the component \texttt{points}. For an exact representation, we can specify
the dimension argument equal to $n-2$.

\begin{Example}\label{ex0}
To illustrate application of Theorem \ref{thmMDS} for non-Euclidean
dissimilarities, we computed the Bray-Curtis dissimilarity matrix of the
\texttt{iris} \emph{setosa} data, a four-dimensional data set available
in R. The Bray-Curtis dissimilarity
defined in Cox and Cox \cite[Table~1.1]{cc2001} as
$$
\delta_{ij} = \frac{1}{p} \frac{\sum_k |x_{ik}
- x_{jk}|}{\sum_k(x_{ik} + x_{jk})},
\qquad x_i, x_j \in \mathbb R^p,
$$
is not a distance since it does not satisfy the triangle inequality.
A Bray-Curtis method is available in the
\texttt{distance} function of the \emph{ecodist} package for R
\cite{ecodist}. We find a configuration of 50
points in $\mathbb R^{48}$ that have $\mathcal U$-centered distances equal to the $\mathcal U$-centered dissimilarities. The MDS
computations are handled by the R function
\texttt{cmdscale}. Function \texttt{Ucenter}, which
implements $\mathcal U$-centering, is in the R package \emph{pdcor} \cite{pdcor}.
\begin{verbatim}
> x <- iris[1:50, 1:4]
> iris.d <- distance(x, method="bray-curtis")
> AU <- Ucenter(iris.d)
> v <- cmdscale(as.dist(AU), k=48, add=TRUE)$points
\end{verbatim}
The points \texttt{v} are a $50 \times 48$ data matrix, of 50 points
in $\mathbb R^{48}$. Next we compare the
$\mathcal U$-centered distance matrix of the points \texttt{v} with the original
object \texttt{AU} from $\mathscr H_n$:
\begin{verbatim}
> VU <- Ucenter(v)
> all.equal(AU, VU)
[1] TRUE
\end{verbatim}
The last line of output shows that the points \texttt{v} returned by \texttt{cmdscale} have $\mathcal U$-centered distance matrix \texttt{VU}
equal to our original element \texttt{AU} of the Hilbert space.
\hfill $\diamond$
\end{Example}

Example \ref{ex0} shows that the sample distance covariance can be
defined for dissimilarities via the inner
product in $\mathscr H_n$. Alternately one can compute
$\mathcal V^2_n(\mathbf U, \mathbf V)$, where $\mathbf U$,
$\mathbf V$ are the Euclidean representations corresponding
to the two $\mathcal U$-centered dissimilarity matrices that
exist by Theorem \ref{thmMDS}. Using the corresponding
definitions of distance variance, sample distance
correlation for dissimilarities is well defined by
(\ref{Rcorrected}) or $\mathcal R^2_n(\mathbf U, \mathbf V)$.
Similarly one can define pdCov and pdCor when one or more of
the dissimilarity matrices of the samples is not
Euclidean distance. However, as in the case of Euclidean distance,
we need to define the corresponding population coefficients, and
develop a test of independence.  For the population definitions
see Section \ref{S4}.

\subsection{Simplified computing formula for pdCor}

Let
\begin{align}\label{Rcorrected}
R^*_{x,y} &:=\left\{
            \begin{array}{ll}
\frac{(\widetilde A \cdot \widetilde B)}{|\widetilde A| |\widetilde B|}, & \hbox{$|\widetilde A| |\widetilde B|
\neq 0$;} \\
              0, & \hbox{$|\widetilde A| |\widetilde B| = 0$,}
            \end{array}
          \right.
\end{align}
where $\widetilde A=\widetilde A(x), \widetilde B=\widetilde B(y)$ are the $\mathcal U$-centered distance matrices of the
samples $x$ and $y$, and $|\widetilde A|=(\widetilde A \cdot \widetilde A)^{1/2}$. The statistics $R^*_{x,y}$ and
$R^*(x, y;z)$ take values in $[-1, 1]$, but they are measured in units comparable to the squared distance
correlation $\mathcal R^2_n(\mathbf X, \mathbf Y)$.

\begin{Proposition}\label{P:pdcor2}
If  $(1-(R^*_{x,z})^2)(1-(R^*_{y,z})^2) \neq 0$, a computing formula for $R^*_{x, y ; z}$ in Definition
(\ref{pdcor2}) is
\begin{align}
R^*_{x, y ; z} =
\frac{R^*_{x, y} - R^*_{x, z}R^*_{y, z}}{\sqrt{1-(R^*_{x,z})^2}{\sqrt{1-(R^*_{y,z})^2}}}. \label{R.alt}
\end{align}

\end{Proposition}
See Appendix \ref{prfRxy} for a proof.

Equation (\ref{R.alt}) provides a simple and familiar form of computing formula for the partial distance
correlation. The computational algorithm is easily implemented, as summarized below.

\emph{Algorithm to compute partial distance correlation $R^*_{x,y ; z}$
from Euclidean distance matrices $A=(|x_i-x_j|)$, $B=(|y_i-y_j|)$, and $C=(|z_i-z_j|)$}:

\begin{enumerate}
\item Compute the $\mathcal U$-centered matrices $\widetilde A$, $\widetilde B$, and $\widetilde C$,
using
\begin{equation*}
\widetilde A_{i,j}=
       a_{i,j} - \frac {a_{i.}}{n-2}  - \frac{a_{.j}}{n-2}
 + \frac{a_{..}}{(n-1)(n-2)}, \qquad i \neq j,
 \end{equation*}
and $\widetilde A_{i,i}=0$.
\item Compute inner products and norms
using $$ (\widetilde A \cdot \widetilde B)=\frac{1}{n(n-3)} \sum_{i \neq j} \widetilde A_{i,j}\widetilde
B_{i,j},
\qquad |\widetilde A|=(\widetilde A \cdot \widetilde A)^{1/2}
$$ and $R^*_{x,y}$, $R^*_{x,z}$, and $R^*_{y,z}$ using $ R^*_{x,y}=\frac{(\widetilde A \cdot \widetilde
B)}{|\widetilde A| |\widetilde B|}.$
\item If $R_{x,z}^2\neq 1$ and $R_{y,z}^2 \neq 1$
$$ R^*_{x, y ; z} =
\frac{R^*_{x, y} - R_{x, z}R^*_{y, z}}{\sqrt{1-(R^*_{x,z})^2}{\sqrt{1-(R^*_{y,z})^2}}},
$$
\end{enumerate}
otherwise apply the definition (\ref{pdcor2}).

Note that it is not necessary to explicitly compute the projections, when (\ref{R.alt}) is applied. The above
algorithm has a straightforward translation to code; see e.g.\ the \emph{pdcor} package
\cite{pdcor} for an implementation in R.

\section{Population Coefficients and Inference\label{S4}}

\subsection{Population coefficients}

The population distance covariance has been defined in terms of the joint
and marginal characteristic
functions of the random vectors. Here we give an equivalent definition following Lyons \cite{lyons2013}, who generalizes distance correlation to separable Hilbert spaces. Instead of starting with the distance matrices
$(a_{ij})=(|x_i-x_j|_p)$ and $(b_{ij})=(|y_i-y_j|_q),$ the starting point of the population definition are the
bivariate distance functions $ a(x, x'):= |x-x'|_p $ and $ b(y,y')=|y- y'|_q,$
where $x, x'$ are realizations of
the random variables $X$ and $y, y'$ are realizations of the random variable $Y$.

We can also consider the random versions. Let $X \in \mathbb R^p$ and $Y \in \mathbb R^q$ be random variables
with finite expectations. The random distance functions are $ a(X, X'):= |X-X'|_p $ and $ b(Y,Y')=|Y- Y'|_q $.
Here the primed random variable $X$ denotes an independent and identically distributed (iid) copy of the variable
$X$, and similarly $Y, Y'$ are iid.

The population operations of double centering involves expected values
with respect to the underlying population random variable.
For a given random variable $X$ with cdf $F_X$, we define the corresponding
\emph{double centering function with respect to} $X$ as
\begin{align} \label{e:HatAx}
A_X(x, x') &:= a(x, x') - \int_{\mathbb R^p} a(x, x') dF_X(x')
   - \int_{\mathbb R^p} a(x, x') dF_X(x)
\\ & \qquad + \int_{\mathbb R^p} \int_{\mathbb R^p} a(x, x') dF_X(x') dF_X(x),
\notag
\end{align}
provided the integrals exist.

Here $A_X(x, x')$ is a real valued function of two realizations
of $X$, and the subscript $X$ references the underlying random variable.
Similarly for $X, X'$ iid with cdf $F_X$, we define the random variable
$ A_X$ as an abbreviation for $ A_X(X, X')$, which is a
random function of $(X, X')$. Similarly we define
\begin{align*}
 B_Y(y, y') &:= b(y, y') - \int_{\mathbb R^q} b(y, y') dF_Y(y')
   - \int_{\mathbb R^q} b(y, y') dF_Y(y)
\\ & \qquad + \int_{\mathbb R^q} \int_{\mathbb R^q} b(y, y') dF_Y(y') dF_Y(y),
\end{align*}
and the random function $ B_Y := B_Y(Y, Y')$.

Now for $ X, X'$ iid, and $Y, Y'$ iid, such that $X, Y$ have finite expectations, the population distance
covariance $\mathcal V(X,Y)$ is defined by
\begin{equation}\label{e:dcov2}
\mathcal V^2(X,Y):= E[ A_X\,  B_Y].
\end{equation}
The definition (\ref{e:dcov2}) of $\mathcal V^2(X, Y)$ is equivalent to the original definition (\ref{e:Aw}).
However, as we will see in the next sections, (\ref{e:dcov2})
is an appropriate starting point to develop the
corresponding definition of pdCov and pdCor population coefficients.

More generally, we can consider {\it dissimilarity functions} $a(x, x')$.
In this paper, a dissimilarity function is a
symmetric function $a(x,x'):\mathbb R^p \times \mathbb R^p \to \mathbb R$ with $a(x,x) = 0$. The corresponding
random dissimilarity functions $a(X, X')$ are random variables such that $a(X, X')=a(X', X)$, and $a(X, X)=0$.

Double centered dissimilarities are formally defined by the same
equations as double centered distances in (\ref{e:HatAx}).

The following lemma establishes that linear combinations of
double-centered dissimilarities are double-centered dissimilarities.

\begin{Lemma}\label{lemma3}
Suppose that $X \in \mathbb R^p$, $Y \in \mathbb R^q$,
$a(x,x')$ is a dissimilarity on $\mathbb R^p \times \mathbb R^p$,
and
$b(y,y')$ is a dissimilarity on $\mathbb R^q \times \mathbb R^q$.
Let $  A_X(x,x')$ and $  B_Y(y,y')$ denote the
dissimilarity obtained by double-centering $a(x,x')$ and $b(y,y')$,
respectively. Then if $c_1$ and $c_2$ are real scalars,
$$
c_1   A_X(x,x')+ c_2   B_Y(y,y') =   D_T(t,t'),
$$
where $T=[X,Y] \in \mathbb R^p \times \mathbb R^q$, and $  D_T(t,t')$
is the result of double-centering $d(t,t')=c_1 a(x,x')+ c_2 b(y,y')$.
\end{Lemma}
See Appendix \ref{prfLemma3} for the proof.

The linear span of double-centered distance functions $  A_X(x,x')$
is a subspace of the space of double-centered dissimilarity functions with
the property $a(x,x') = O(|x| + |x'|)$. In this case all integrals in
(\ref{e:HatAx}) and $  A_X$ are finite if $X$ has finite expectation.
The linear span of the random functions $  A_X$ for random
vectors $X$ with $E|X| < \infty$ is clearly a linear space, such that the
linear extension of (\ref{e:dcov2})
$$
E[  A_X   B_Y] = \mathcal V^2(X,Y)
$$
to the linear span is an inner product space or pre-Hilbert space;
its completion with respect to the
metric arising from its inner product
\begin{equation}\label{e:IP}
(  A_X \cdot   B_Y) :=
E[  A_X   B_Y]
\end{equation}
(and norm) is a Hilbert space
which we denote by $\mathscr H$.

\begin{Definition}[Population partial distance covariance]\label{defPDCOV}
Introduce the scalar coefficients $$
\alpha := \frac{\mathcal V^2(X,Z)}{\mathcal V^2(Z,Z)}, \qquad
\beta := \frac{\mathcal V^2(Y,Z)}{\mathcal V^2(Z,Z)}.
$$ If $\mathcal V^2(Z,Z)=0$ define $\alpha=\beta=0$. The double-centered
projections of $  A_X$ and
$  B_Y$ onto the orthogonal complement of  $  C_Z$
in Hilbert space $\mathscr H$ are defined
$$
  P_{Z^\perp}(X):=  A_X(X, X')- \alpha   C_Z(Z, Z'), \quad
   P_{Z^\perp}(Y):=  B_Y(Y, Y')- \beta   C_Z(Z, Z'),
$$ or in short $P_{Z^\perp}(X)=  A_X - \alpha   C_{Z}$
and $P_{Z^\perp}(Y)=  B_Y -
\beta   C_Z$, where $  C_Z$ denotes double-centered with
respect to the random variable $Z$.

The population partial
distance covariance is defined by the inner product
$$
(P_{Z^\perp}(X) \cdot P_{Z^\perp}(Y)) := E[(
A_X- \alpha   C_{Z}) \cdot
(  B_Y- \beta   C_{Z})].
$$
\end{Definition}

\begin{Definition}[Population pdCor]\label{defPDCOR}
Population partial distance correlation is defined
\begin{align*}
\mathcal R^*(X, Y;Z) :=
  \frac{(P_{Z^\perp}(X) \cdot P_{Z^\perp}(Y))}
{|P_{Z^\perp}(X)||P_{Z^\perp}(Y)|},
\end{align*}
where $|P_{Z^\perp}(X)|=(P_{Z^\perp}(X) \cdot P_{Z^\perp}(X))^{1/2}.$
If $|P_{Z^\perp}(X)||P_{Z^\perp}(Y)| = 0$
we define $\mathcal R^*(X, Y;Z) = 0$.
\end{Definition}

Note that if $\alpha=\beta=0$, we have $ (P_{Z^\perp}(X) \cdot P_{Z^\perp}(Y))
= E[  A_X \cdot   B_Y] =
\mathcal V^2(X,Y), $ and  $\mathcal R^*(X, Y; Z) = \mathcal R^*(X, Y)$.

The population coefficient of partial distance correlation can
be evaluated in terms of the pairwise distance correlations using
Formula (\ref{pdCor.pop}) below. Theorem \ref{T:pdCor.pop} establishes
that (\ref{pdCor.pop}) is equivalent to Definition \ref{defPDCOR},
and therefore serves as an alternate definition of population
pdCor.

\begin{thm}[Population pdCor]\label{T:pdCor.pop}
The following definition of population partial distance correlation
is equivalent to Definition \ref{defPDCOR}.
\begin{align}\label{pdCor.pop}
& \mathcal R^*(X,Y;Z) = \\ & \quad
\left\{
  \begin{array}{ll}
 \frac{\mathcal R^2(X,Y) - \mathcal R^2(X,Z) \mathcal R^2(Y,Z)}
{\sqrt{1-\mathcal R^4(X,Z)} \sqrt{1-\mathcal R^4(Y,Z)}},
 & \hbox{$\mathcal R(X,Z)\neq 1$ \textrm{and} $\mathcal R(Y,Z)\neq 1$;} \\
    0, & \hbox{$\mathcal R(X,Z)=1$ \textrm{or} $\mathcal R(Y,Z)=1$.}
  \end{array}
\right.
\notag
\end{align}
where $\mathcal R(X, Y)$ denotes the population distance correlation.
\end{thm}

The proof of Theorem \ref{T:pdCor.pop} is given in Appendix \ref{prfThm3}

\subsection{Discussion}

Partial distance correlation is a scalar quantity that captures
dependence, while conditional inference (conditional distance
correlation) is a more complex notion as it is a function of the condition.
Although one might hope that partial distance correlation
$\mathcal R^*(X,Y;Z) =0$ if and only if $X$, and $Y$ are
conditionally independent given $Z$, this is, however, not the case.
Both notions capture \emph{overlapping} aspects of dependence,
but mathematically they are not equivalent. The following examples
illustrate that $pdCor(X,Y;Z)=0$
is not equivalent to conditional independence of $X$ and $Y$ given $Z$.

Suppose that $Z_1, Z_2, Z_3$ are iid standard normal variables,
$X=Z_1+Z_3$, $Y=Z_2+Z_3$, and
$Z=Z_3$. Then $\rho(X,Y)=1/2$, $\rho(X,Z)=\rho(Y,Z)=1/\sqrt{2}$,
$pcor(X,Y;Z)=0$, and $X$ and $Y$ are conditionally independent of $Z$.
One can evaluate $\pdCor(X, Y;Z)$ by applying (\ref{pdCor.pop}) and the
following result. Suppose that $(X,Y)$ are jointly bivariate normal with
correlation $\rho$.
Then by \cite[Theorem 7(ii)]{srb07},
\begin{equation}\label{cor2dcor}
\mathcal R^2(X,Y) = \frac{\rho\, \textrm{arcsin} \rho +
  \sqrt{1-\rho^2} - \rho\, \textrm{arcsin}(\rho/2)
  - \sqrt(4-\rho^2) + 1}{1 + \pi/2 - \sqrt{3}}.
\end{equation}
In this example $\mathcal R^2(X,Y)=0.2062$,
$\mathcal R^2(X,Z)=\mathcal R^2(Y,Z)=  0.4319 $, and $\pdCor(X,Y;Z)= 0.0242$.

In the other direction we can construct a trivariate normal $(X, Y, Z)$
such that $\mathcal R^*(X,Y;Z)=0$ but $pcor(X,Y;Z) \neq 0$. According to
Baba, Shibata and Sibuya \cite{bss2004}, if $(X,Y,Z)$ are trivariate
normal, then pcor$(X,Y;Z) = 0$ if and only if $X$ and $Y$ are
conditionally independent given $Z$. A specific numerical example is
as follows. By inverting
equation (\ref{cor2dcor}), given any $\mathcal R^2(X,Y)$,
one can solve for $\rho(X,Y)$ in $[0,1]$. When $\mathcal R(X,Y)=0.04$,
and $\mathcal R(X,Z)=\mathcal R(Y,Z)=0.2$,
the nonnegative solutions are
$\rho(X,Y)=0.22372287$, $\rho(X,Z)=\rho(Y,Z)= 0.49268911$. The corresponding
$3 \times 3$ correlation matrix is positive definite. Let $(X,Y,Z)$
be trivariate normal
with zero means, unit variances, and correlations $\rho(X,Y)=0.22372287$,
$\rho(X,Z)=\rho(Y,Z)= 0.49268911$.
Then pcor$(X,Y;Z)=$ $-0.025116547$
and therefore conditional independence of $X$ and $Y$
given $Z$ does not hold, while $\mathcal R^*(X,Y;Z)$ is exactly zero.

In the non-Gaussian case, it is not true that zero partial correlation
is equivalent to conditional independence, while partial distance correlation
has the advantages that it can capture non-linear dependence, and is
applicable to vector valued random variables.

\subsection{Inference}

We have defined partial distance correlation, which is bias corrected and has an informal interpretation, but we
also want to determine whether $R^*(x,y;z)$ is significantly different from zero, which is not a trivial problem.
Although theory and implementation for a consistent test of multivariate independence based on distance
covariance is available \cite{srb07, energy}, clearly the implementation of the test by randomization is not
directly applicable to the problem of testing for zero partial distance correlation.

Since $\pdCor(X,Y;Z)=0$ if and only if $\pdCov(X,Y;Z)=0$, we develop a test for $H_0: \pdCov(X,Y;Z)=0$ vs $H_1:
\pdCov(X,Y;Z)\neq 0$ based on the inner product $(P_{z^\perp}(x) \cdot P_{z^\perp}(y))$. Unlike the problem of
testing independence of $X$ and $Y$, however, we do not have the distance matrices or samples corresponding to
the projections $P_{z^\perp}(x)$ and $P_{z^\perp}(y)$. For a test we need to consider that $P_{z^\perp}(x)$ and
$P_{z^\perp}(y)$ are arbitrary elements of $\mathscr H_n$.

We have proved (Theorem \ref{thmMDS}) that $P_z(x)$ can be represented as the $\mathcal U$-centered distance matrix of
some configuration of $n$ points $\mathbf U$ in a Euclidean space $\mathbb R^p$, $p \leq n-2$, and that $P_z(y)$ has
such a representation $\mathbf V$ in $\mathbb R^q$, $q \leq n-2$. Hence, a test for $\pdCov(X,Y;Z)=0$ can be defined
by applying the distance covariance test statistic $\mathcal V^2_n(\mathbf U, \mathbf V)$. The resulting test
procedure is practical to apply, with similar computational complexity as the original dCov test of independence.
One can apply the test of independence implemented by permutation bootstrap in the
\texttt{dcov.test} function of the energy package \cite{energy} for R. Alternately
a test based on the inner product (\ref{PDCOV}) is implemented in the \emph{pdcor} \cite{pdcor} package.

\section{Examples and Applications}

In this section we summarize simulation results for tests of the
null hypothesis of zero partial distance correlation and two applications.
We compared our simulation results with two types of tests based on linear correlation. Methods for dissimilarities are illustrated using data on
genetic distances for samples of maize in Example \ref{ex6}, and variable
selection is discussed in Example \ref{ex7}.

\subsection{Partial distance covariance test}
The test for zero partial distance correlation is a test of whether the inner product $(P_{Z^\perp}(X) \cdot
P_{Z^\perp}(Y))$ of the projections is zero. We obtain the corresponding samples $\mathbf U, \mathbf V$ such that
the $\mathcal U$-centered distances of $\mathbf U$ and $\mathbf V$ are identical to $P_z(x)$ and $P_z(y)$, respectively,
using classical (metric) MDS. We used the \texttt{cmdscale} function in R to obtain the samples $\mathbf U$ and
$\mathbf V$ in both cases. The \emph{pdcov} test applies the inner product of the double centered distance
matrices of $\mathbf U$ and $\mathbf V$ as the test statistic. Alternately, one could apply the original
\emph{dcov} test, to the joint sample $(\mathbf U, \mathbf V)$. These two test statistics are essentially
equivalent, but the \emph{pdcov} test applies an unbiased statistic defined by (\ref{VU}), while the \emph{dcov}
test applies a biased statistic, $n \mathcal V_n^2(\mathbf U, \mathbf V)$ (\ref{e:AnXY}). Both tests are
implemented as permutation tests, where for a test of $\pdCor(X, Y; Z)=0$ the sample indices of the $X$ sample
are randomized for each replicate to obtain the sampling distribution of the test statistic under the null
hypothesis.

The term `permutation test' is sometimes restricted to refer to the exact test, which is implemented by
generating all possible permutations. Except for very small samples, it is not computationally feasible to
generate all permutations, so a large number of randomly generated permutations are used. This approach, which we
have implemented, is sometimes called a randomization test to distinguish it from an exact permutation test.

In our simulations the \emph{pdcov} and \emph{dcov} tests for zero partial distance correlation were equivalent
under null or alternative hypotheses
 in the sense that the type 1 error rates and estimated power agreed to
within one standard error. In power comparisons, therefore, we reported only the \emph{pdcov} result.

\subsection{Partial correlation}

The linear partial correlation $r(x, y; z)$ measures the partial correlation between one dimensional data vectors
$x$ and $y$ with $z$ removed (or controlling for $z$). The sample partial correlation coefficient is
\begin{equation}\label{pcor.stat}
 r(x, y; z) = \frac{r(x, y) - r(x, z) r(y, z)}{\sqrt{(1-r(x,z)^2)(1-r(y,z)^2)}},
\end{equation}
where $r(x, y)$ denotes the linear (Pearson) sample correlation. The partial correlation test is usually implemented as a $t$ test (see e.g.\ Legendre
\cite[p. 40]{legendre2000}). In examples where $x$, $y$, $z$ are one
dimensional, we have included the partial correlation $t$-test in comparisons. However, for small samples
and some non-Gaussian data, the type 1 error rate is inflated.
In cases where type 1 error rate of \emph{pcor} was not controlled we did not
report power.

\subsection{The partial Mantel test}

The partial Mantel test is a test of the hypothesis that there is a linear association between the distances of
$X$ and $Y$, controlling for $Z$. This extension of the Mantel test \cite{mantel1967} was proposed by Smouse, et
al. \cite{sls1986} for a partial correlation analysis on three distance matrices. The Mantel and partial Mantel
tests are commonly applied in community ecology (see e.g.\ Legendre and Legendre
\cite{ll2012}), population genetics, sociology, etc.

Let $U_1,U_2,U_3$ denote the upper triangles of the $n \times n$ distance matrices of the $X$, $Y$ and $Z$
samples, respectively. Then the partial Mantel test statistic is the sample linear partial correlation between
the $n(n-1)/2$ elements of $U_1$ and $U_2$ controlling for $U_3$. Let $u_1, u_2, u_3$ be the corresponding data
vectors obtained by representing $U_1,U_2,U_3$ as vectors. The partial Mantel statistic is $r(u_1,u_2;u_3)$,
computed using formula (\ref{pcor.stat}).

Since $u_i$ are not iid samples, the usual $t$ test is not applicable, so the partial Mantel test is usually
implemented as a permutation (randomization) test. See Legendre \cite{legendre2000} for a detailed algorithm and
simulation study comparing different methods of computing a partial Mantel test statistic. Based on the results
reported by Legendre, we implemented the method of permutation of the raw data. The algorithm is given in detail
on page 44 by Legendre \cite{legendre2000}, and it is very similar to the algorithm we have applied for the
\emph{energy} tests (\emph{pdcov} and \emph{dcov} tests for zero pdCor). This method (permutation of raw data)
for the partial Mantel test is implemented in the \emph{ecodist} package \cite{ecodist} and also the \emph{vegan}
package
\cite{vegan} for R.
For simulations, the \texttt{mantel} function in the \emph{ecodist} package, which is implemented mainly in
compiled code, is much faster than the \texttt{mantel} or \texttt{mantel.partial} functions in the \emph{vegan}
package, which are implemented in R.

\begin{remark}\label{covdist}
Both the \emph{pdcov} and partial Mantel (\emph{Mantel}) tests are based on distances. One may ask ``is distance
covariance different or more general than covariance of distances?'' The answer is yes; it can be shown that $$
 \dCov^2(X, Y) = \Cov(|X-X'|, |Y-Y'|) - 2 \Cov(|X-X'|, |Y-Y''|),
$$ where $(X,Y), (X',Y'),$ and $(X'',Y'')$ are iid. The \emph{dcov} tests are tests of independence of $X$ and
$Y$ ($\dCov^2(X, Y)=0$), while the Mantel test is a test of the hypothesis $\Cov(|X-X'|, |Y-Y'|)=0$. An example
of dependent data such that their distances are uncorrelated but $\dCov^2(X, Y) > 0$ is given e.g. by Lyons \cite{lyons2013}. Thus, distance covariance tests are more general than Mantel tests, in the sense that distance
covariance measures all types of departures from independence.
\end{remark}

\subsection{Simulation design}
In each permutation test $R=999$ replicates are generated and the estimated $p$-value is computed as $$
 \widehat p = \frac{1 + \sum_{k=1}^R I(T^{(k)} \geq T_0)}{1 + R},
$$ where $I(\cdot)$ is the indicator function, $T_0$ is the observed value of the test statistic, and $T^{(k)}$
is the statistic for the $k$-th sample. The test is rejected at significance level $\alpha$ if $\widehat p \leq
\alpha$. The partial correlation test (\emph{pcor}) is also included for comparison. It is implemented as a
t-test \cite{ppcor}.

Type 1 error rates and estimated power are determined by a simulation size of 10,000 tests in each case; for
$n=10$ the number of tests is 100,000. The standard error is at most 0.005 (0.0016 for $n=10$).

\begin{Example}\label{ex1a}
This example is a comparison of type 1 error on uncorrelated standard normal data. The vectors $X$, $Y$, and $Z$
are each iid standard normal. Results summarized in Table \ref{T:ex1a} show that type 1 error rates for
\emph{pdcov}, \emph{dcov}, and \emph{partial Mantel} tests are within two se of the nominal significance level,
while type 1 error rates for pcor are inflated for $n=10, 20$.
\end{Example}

\begin{table}[ht]
\caption{Example \ref{ex1a}: Type 1 error rates at nominal significance level
$\alpha$ for uncorrelated standard trivariate normal data.\label{T:ex1a}}
\centering
\begin{tabular}{|r|r|rrrr|r|rrrr|}
  \hline
$n$ & $\alpha$ &  pdcov & dcov & Mantel & pcor & $\alpha$&  pdcov & dcov & Mantel & pcor \\
  \hline
  10 & 0.05 &     0.051 & 0.051 & 0.050 & 0.090 &  0.10 &   0.101 & 0.102 & 0.100 & 0.143 \\
  20 & 0.05 &     0.048 & 0.048 & 0.051 & 0.064 &  0.10 &   0.098 & 0.100 & 0.100 & 0.120 \\
  30 & 0.05 &     0.052 & 0.051 & 0.047 & 0.059 &  0.10 &   0.095 & 0.097 & 0.098 & 0.107 \\
  50 & 0.05 &     0.051 & 0.051 & 0.051 & 0.058 &  0.10 &   0.104 & 0.104 & 0.106 & 0.105 \\
 100 & 0.05 &     0.049 & 0.050 & 0.048 & 0.057 &  0.10 &   0.100 & 0.100 & 0.096 & 0.105 \\
   \hline
\end{tabular}
\end{table}

\begin{Example}\label{ex1b}
This example is the same as Example \ref{ex1a} except that the $X$ variable is standard lognormal rather than
standard normal. Results summarized in Table \ref{T:ex1b} demonstrate that the type 1 error rates are controlled
at their nominal significance level for \emph{pdcov}, \emph{dcov}, and \emph{partial Mantel} tests (all
implemented as permutation tests), while the \emph{pcor} $t$ test has inflated type 1 error rates for $n \leq
30$.
\end{Example}

\begin{table}[ht]
\caption{Example \ref{ex1b}: Type 1 error rates at nominal significance level
$\alpha$ for uncorrelated non-normal data. The partial distance correlation and other statistics are computed for
standard lognormal $X$ and standard normal $Y$, removing standard normal $Z$.\label{T:ex1b}}
\centering
\begin{tabular}{|r|r|rrrr|r|rrrr|}
  \hline
$n$ & $\alpha$ &  pdcov & dcov & Mantel & pcor & $\alpha$&  pdcov & dcov & Mantel & pcor \\
\hline
  10  & 0.05  &0.052 &0.053 &0.050 &0.090 & 0.10 & 0.103 &0.103 &0.099 &0.142  \\
  20  & 0.05  &0.049 &0.050 &0.054 &0.071 & 0.10 & 0.101 &0.102 &0.102 &0.122  \\
  30  & 0.05  &0.049 &0.049 &0.047 &0.062 & 0.10 & 0.104 &0.104 &0.099 &0.112  \\
  50  & 0.05  &0.050 &0.051 &0.047 &0.058 & 0.10 & 0.104 &0.104 &0.097 &0.108  \\
 100  & 0.05  &0.051 &0.052 &0.048 &0.054 & 0.10 & 0.101 &0.101 &0.097 &0.106  \\
 \hline
\end{tabular}
\end{table}

%
%\begin{figure}[ht]
% \begin{center}
% \subfigure[\label{F:ex1a}]{\epsfig{file=type1-normal, width=.49\linewidth}}
% \subfigure[\label{F:ex1b}]{\epsfig{file=type1-logn, width=.49\linewidth}}
% \end{center}
% \caption{
% Type 1 error rates for partial distance covariance, inner product (dcov) test,
% partial Mantel test, and partial correlation test,
% at significance level $\alpha=0.10$.
% Figure (a) summarizes Example \ref{ex1a} (uncorrelated standard normal data).
% Figure (b) summarizes Example \ref{ex1b} (uncorrelated non-normal data).}
%\end{figure}

\begin{Example}\label{ex4}
In this example, power of tests is compared for correlated trivariate normal data with standard normal marginal
distributions. It is a modification of Example \ref{ex1a} such that the variables $X$, $Y$, and $Z$ are each
\emph{correlated} standard normal. The pairwise correlations are $\rho(X,Y)=\rho(X,Z)=\rho(Y,Z)=0.5$. The power
comparison summarized in Figure
\ref{F:ex4} shows that \emph{pdcor} has higher power than \emph{pcor}
or \emph{partial Mantel} tests. The simulation parameters are identical to those described for type 1 error
simulations.
\end{Example}

\begin{Example}\label{ex5}
This example presents a power comparison for correlated non-normal data. It is a modification of Example
\ref{ex1b} such that the variables $X$, $Y$, and $Z$ are each \emph{correlated}, $X$ is standard lognormal, while
$Y$ and $Z$ are each standard normal. The pairwise correlations are $\rho(\log X,Y)=\rho(\log
X,Z)=\rho(Y,Z)=0.5$. The power comparison summarized in Figure
\ref{F:ex5} shows that \emph{pdcor} has higher power than \emph{pcor}
or \emph{partial Mantel} tests. The simulation parameters are identical to those described for type 1 error
simulations. Again the \emph{pdcov} test has superior power performance. The relative performance of the
\emph{pcor} and \emph{partial Mantel} tests are reversed (partial Mantel test with lowest power) in this example
compared with Example \ref{ex4}.
\end{Example}

\begin{figure}[ht]
 \begin{center}
 \subfigure[\label{F:ex4}]{\epsfig{file=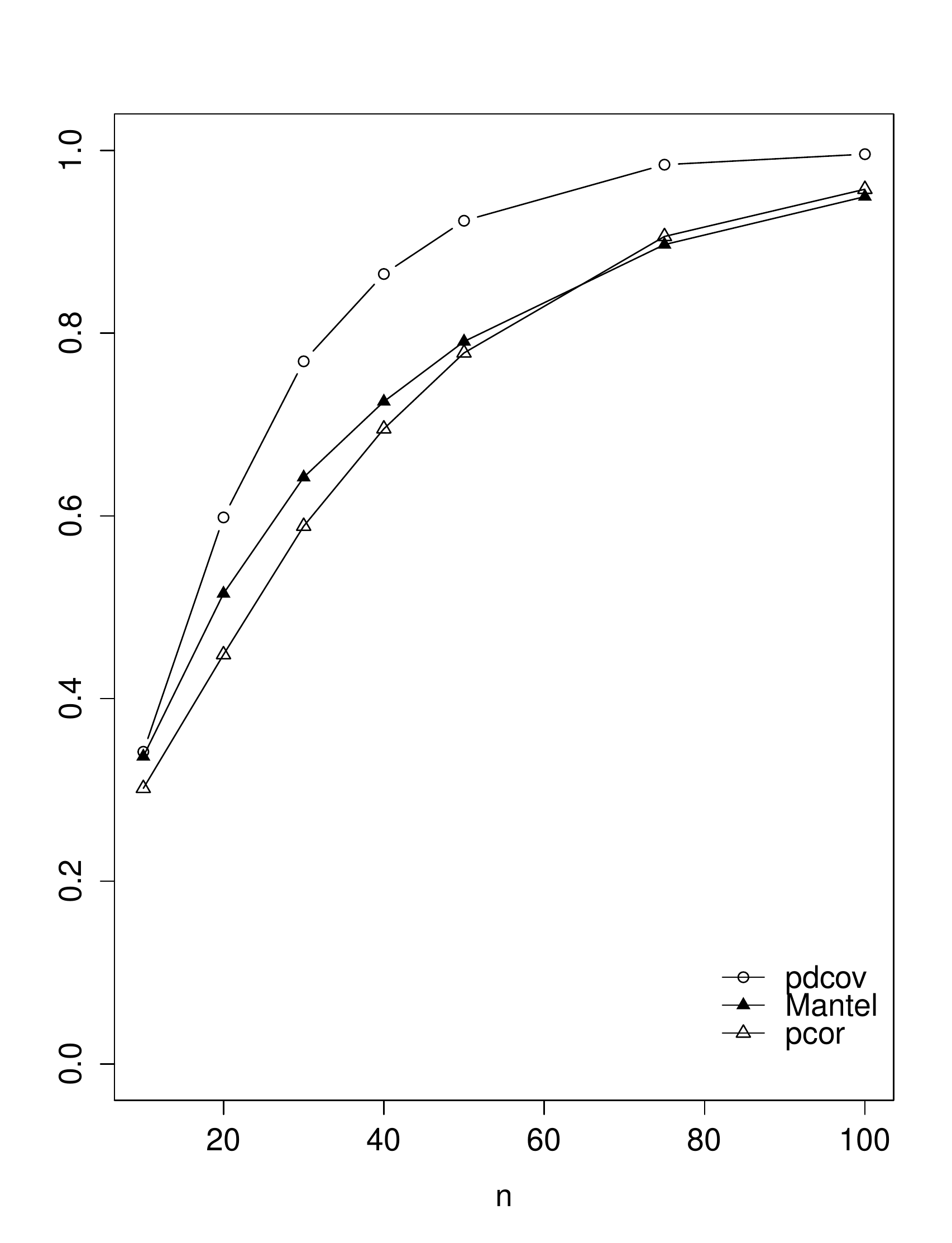, width=.49\linewidth, height=2in}}
 \subfigure[\label{F:ex5}]{\epsfig{file=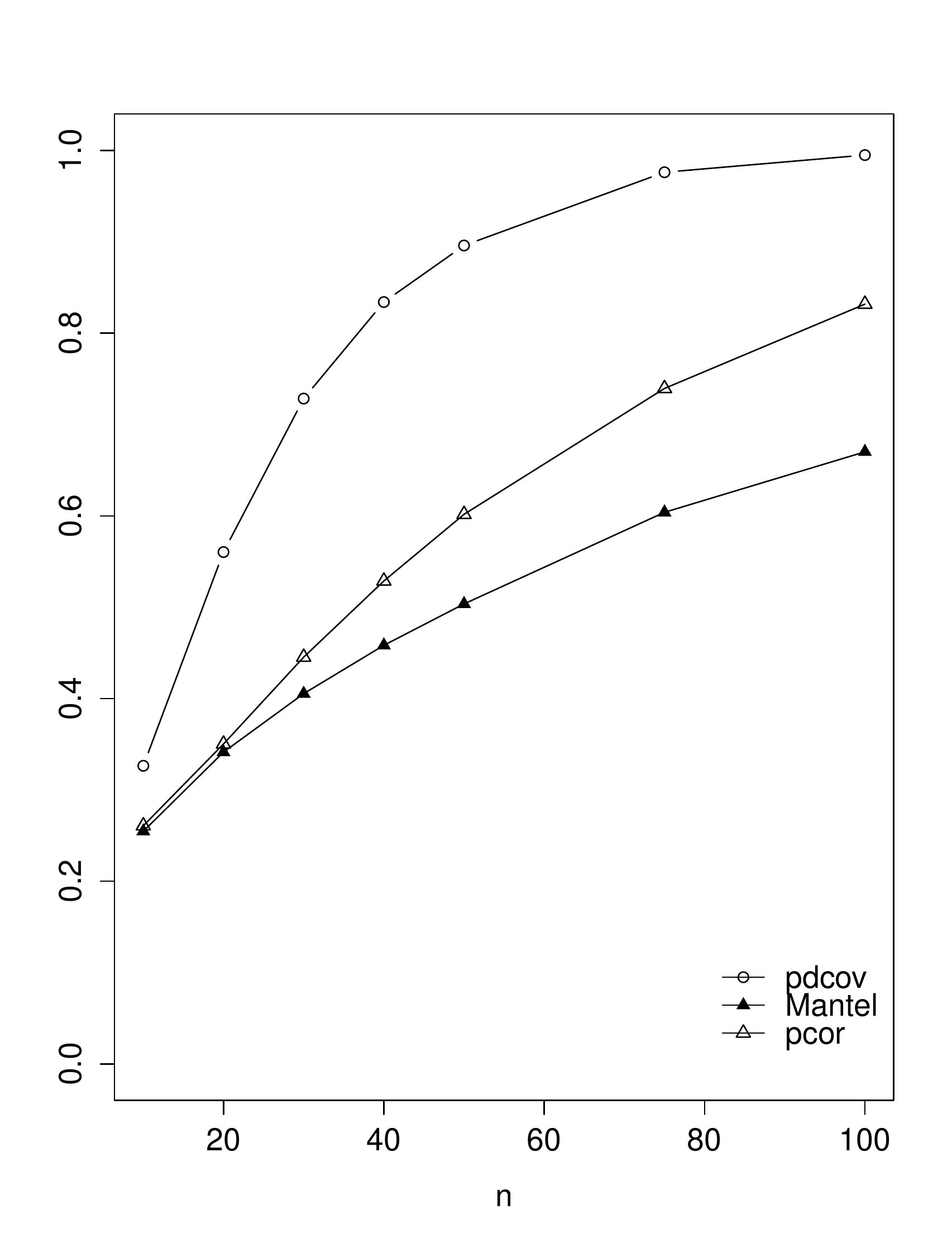, width=.49\linewidth, height=2in}}
 \end{center}
 \caption{
 Power comparisons for partial distance covariance,
 partial Mantel test, and partial correlation test at significance level $\alpha=0.10$.
 Figure (a) summarizes Example \ref{ex4} (correlated standard normal data).
 Figure (b) summarizes Example \ref{ex5} (correlated non-normal data).}
\end{figure}

\begin{Example}\label{ex6}
In this example, we compute distance correlation for dissimilarities and apply the distance covariance inner
product test. The data are two sets of dissimilarities for seven maize populations from Piepho \cite{piepho},
originally from Reif et al.\ \cite{reif}. The two dissimilarity arrays are given in Table \ref{T:maize7}, with genetic distance by simple sequence repeats,
modified Roger's distances below the diagonal and mid-parent heterosis for days to silking of crosses above the
diagonal. Note: comparing with Reif et al. \cite[Tables~2, 4]{reif}, there is an error  in row 2, column 5 of
Piepho \cite[Table 2]{piepho}, where the coefficient for (Pop21, Pop29) should be 0.4, not -0.4.

\begin{table}[h]
\caption{Dissimilarities for Example \ref{ex6}:
Two genetic distances for seven maize populations from Reif et al. \cite{reif}.
\label{T:maize7}}
\centering
\begin{tabular}{rrrrrrrr}
  \hline
 & Pool24 & Pop21 & Pop22 & Pop25 & Pop29 & Pop32 & Pop43 \\
  \hline
  Pool24 &  & 0.50 & -0.40 & 0.70 & -0.30 & -0.70 & -1.30 \\
  Pop21  & 0.22 &  & -0.40 & -0.40 & 0.40 & -0.70 & -1.20 \\
  Pop22 & 0.20 & 0.22 &  & -0.60 & -1.50 & -1.20 & -1.80 \\
  Pop25 & 0.22 & 0.27 & 0.25 &  & -0.90 & -0.90 & -0.50 \\
  Pop29 & 0.22 & 0.24 & 0.23 & 0.26 &  & -0.70 & -0.20 \\
  Pop32 & 0.27 & 0.30 & 0.28 & 0.26 & 0.28 &  & -0.90 \\
  Pop43 & 0.25 & 0.29 & 0.27 & 0.28 & 0.27 & 0.32 &  \\
   \hline
\end{tabular}
\end{table}

Piepho \cite{piepho} compared two tests for association between the matrices; the two-sided $t$-test based on
Pearson's correlation of distances, and the permutation Mantel test with 100,000 permutations. Piepho found that
the tests lead to different decisions. The Pearson correlation $r=-0.44$ had $p$-value 0.0466 for a two-sided
$t$-test, which is significant at the 5\% level. The permutation Mantel test gave a non-significant $p$-value of
0.1446. The different conclusions may be due to fact that the $t$ test applied for the Mantel statistic is on the
liberal side.

We repeated the analysis and also use the data to illustrate the application of dCor and the inner product test
to dissimilarities. Using the
\texttt{cor.test} function in R we compute sample correlation
$r = -0.4388$ and a two-sided $p$-value of .04662. For the permutation Mantel test, we used the
\emph{ecodist} implementation
of Mantel test. The three reported $p$-values are for the alternative hypotheses $\rho \leq 0$, $\rho \geq 0$,
and $\rho \neq 0$, respectively, so for a two tailed Mantel test we have $p$-value 0.14370.
\begin{verbatim}
 ecodist::mantel(dy ~ dx, nperm = 1e+05, nboot = 0)
 ##  mantelr    pval1    pval2    pval3     llim     ulim
 ## -0.43875  0.92562  0.07486  0.14370  0.00000  0.00000
\end{verbatim}

The distance covariance inner product test can be applied to dissimilarities following the method that we
developed for partial distance covariance. Starting with the two dissimilarities which we stored in matrices
\texttt{d1} and \texttt{d2}, we $\mathcal U$-center them, obtain the two representations in Euclidean space, and apply the
inner product test. If we need only the distance correlation statistic $R^*_{xy}$, this can be computed by the
\texttt{Rxy} function in \emph{pdcor} package, which implements formula (\ref{Rcorrected}).
\begin{verbatim}
 AU <- Ucenter.mat(d1)
 BU <- Ucenter.mat(d2)
 u <- cmdscale(as.dist(AU), add = TRUE, k = 5)$points
 v <- cmdscale(as.dist(BU), add = TRUE, k = 5)$points

 Rxy(u, v)
 [1] -0.2734559
\end{verbatim}
The function \texttt{dcovIP.test} in the \emph{pdcor} package \cite{pdcor} for R implements the inner product
dCov test. Alternately one could apply the original dCov test to the points in the Euclidean space.
\begin{verbatim}
 dcovIP.test(u, v, R = 9999)
 ##
 ## dCov inner product test of independence
 ##
 ## data:  replicates 9999
 ## n * V^* = -0.0124, p-value = 0.832
 ## sample estimates:
 ##     R^*
 ## -0.2735
\end{verbatim}
The dCov inner product test is not significant, with $p$-value 0.832.

This example illustrates that the distance correlation or distance covariance in the Hilbert space $\mathscr H_n$
is readily applied to measuring or testing dependence between samples represented by dissimilarities. It is more
general than a test of association such as the Mantel test, as we have shown in Remark \ref{covdist}.
\end{Example}

\begin{Example}[Variable selection]\label{ex7}
This example considers the prostate cancer data from a study by Stamey et al. \cite{prostatepaper}, and is discussed in Hastie, Tibshirani, and Friedman \cite[Ch.~3]{htf2009} in the context of variable selection. The data is from men who were about to have a radical prostatectomy. The response variable \emph{lpsa} measures log PSA (log of the level of prostate-specific antigen). The predictor variables under consideration are eight clinical measures:

\begin{tabular}{ll}
lcavol& log cancer volume \\
lweight& log prostate weight\\
age& age                      \\
lbph& log of the amount of benign prostatic hyperplasia\\
svi& seminal vesicle invasion            \\
lcp& log of capsular penetration\\
gleason& Gleason score            \\
pgg45& percent of Gleason scores 4 or 5\\
\end{tabular}

Here the goal was to fit a linear model to predict
the response \emph{lpsa} given
one or more of the predictor variables above.
The train/test set indicator is in the last column of the data set.
For comparison with the discussion and analysis in Hastie et al. \cite{htf2009}, we
standardized each variable, and used the training set of 67 observations
for variable selection.

Feature screening by distance correlation has been investigated by
Li et al. \cite{lzz2012}. In this example we introduce a partial distance
correlation criterion for variable selection. For simplicity, we implement
a simple variant of forward selection.
This criterion, when applied for a linear model, can help to identify
possible important variables that have strong non-linear association with
the response, and thus help researchers to improve a linear model by
transforming variables or improve prediction by extending to a
nonlinear model.

In the initial step of pdCor forward selection,
the first variable to enter the model is the variable
$x_j$ for which distance correlation $R_{x_j,y}$ with response $y$ is largest.
After the initial step, we have a model with one predictor $x_j$, and we compute
pdCor$(y, x_k; x_j)$, for the variables $x_k \neq x_j$ not in the model, then
select the variable $x_k$ for which pdCor$(y, x_k; x_j)$ is largest.
Then continue, at each step computing $\pdCor(y, x_j;w)$ for every
variable $x_j$ not yet in the model, where $w$ is the vector of
predictors currently in the model. The variable to enter next is the one
that maximizes $\pdCor(y, x_j; w)$.

According to the pdCor criterion, the variables selected enter the model
in the order: \emph{  lcavol, lweight, svi, gleason, lbph, pgg45}.

If we set a stopping rule at 5\% significance for the
pdCor coefficient, then we stop after adding \emph{lbph}
(or possibly after adding \emph{pgg45}).
At the step where \emph{gleason} and \emph{lbph} enter,
the $p$-values for significance of pdCor are 0.016 and 0.001, respectively. The
$p$-value for \emph{pgg45} is approximately 0.05.

The models selected by pdCor, best subset method, and lasso (Hastie et al.
\cite[Table 3.3, p. 63]{htf2009}) are:
$$
%\left\{
  \begin{array}{ll}
    \textrm{pdCor:} & \hbox{lpsa $\sim$ lcavol + lweight + svi + gleason;} \\
    \textrm{best subsets:} & \hbox{lpsa $\sim$ lcavol + lweight ;} \\
    \textrm{lasso:} & \hbox{lpsa $\sim$ lcavol + lweight + svi + lbph.}
  \end{array}
%\right.
$$
The order of selection for ordinary forward stepwise selection (Cp) is
\emph{ lcavol, lweight, svi, lbph, pgg45, age, lcp, gleason}.

Comparing pdCor forward selection with lasso and forward stepwise selection,
we see that the results are similar, but \emph{gleason} is not in the
lasso model and enters last in the forward stepwise selection, while it
is the fourth variable to enter the pdCor selected model.  The raw
Gleason Score is an integer from 2 to 10 which
is used to measure how aggressive is the tumor, based on a prostate biopsy.
Plotting the
data (see Figure \ref{F:prostate1}) we can observe that
there is a strong \emph{non-linear} relation between \emph{gleason} and
\emph{lpsa}.

This example illustrates that partial distance correlation has practical
use in variable selection and in model checking. If we were using
pdCor only to check the model selected by another procedure, it would
show in this case that there is some nonlinear dependence remaining
between the response and the predictors excluded from the lasso model
or the best subsets model. Using the
pdCor selection approach, we also learn which of the remaining
predictors may be important.

\begin{figure}[ht]
 \begin{center}
 \epsfig{file=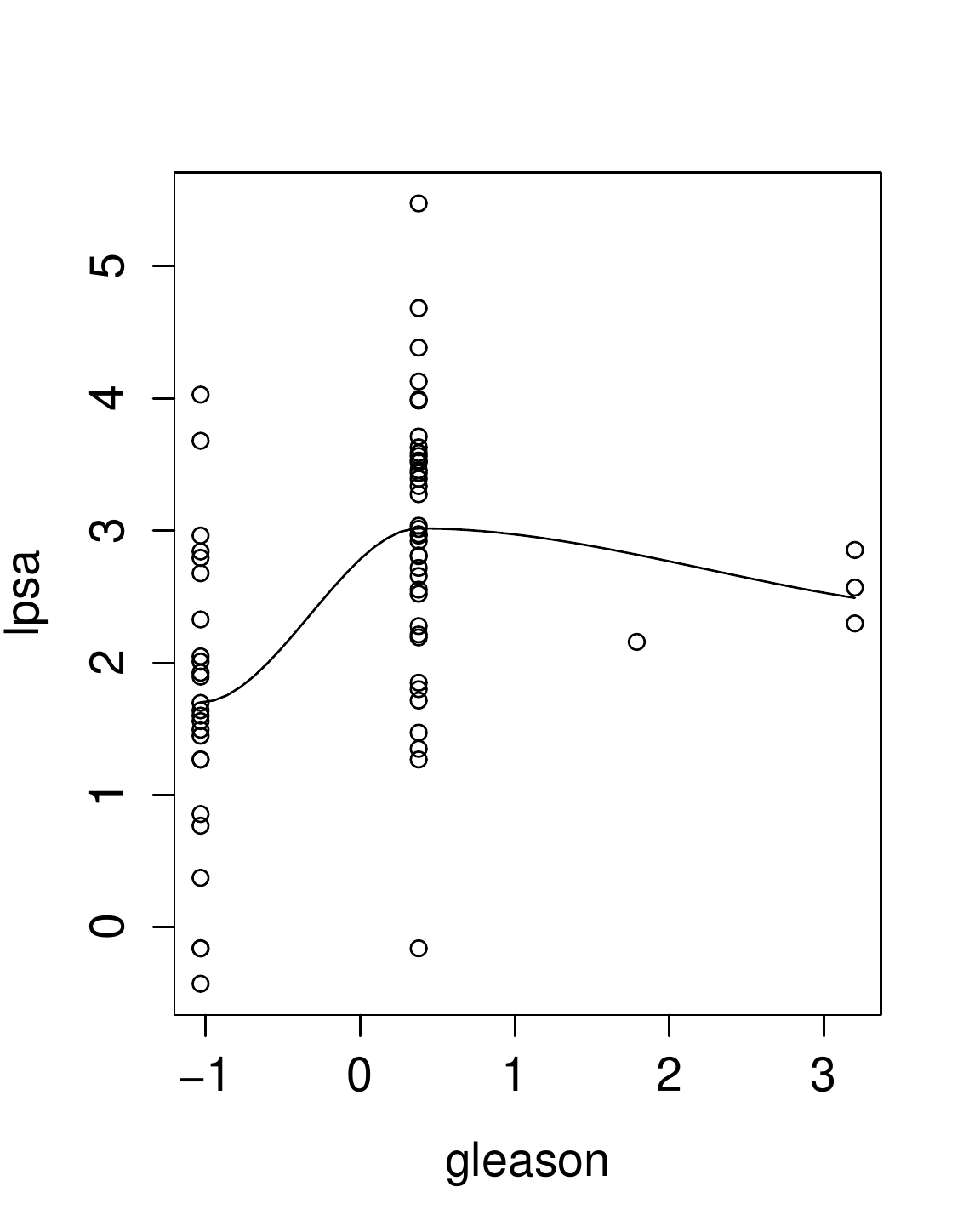, width=.5\linewidth,height=2.25in}
 \end{center}
 \caption{Scatter plot of response lpsa vs Gleason score with a loess smoother.
\label{F:prostate1}}
\end{figure}

Finally, it is useful to recall that pdCor has more flexibility to
handle predictors that are multi-dimensional. One may want groups
of variables to enter or leave the model as a set.
It is often the case that when dimension of the feature space is high,
many of the predictor variables are highly correlated. In this case,
methods such as Partial Least Squares are sometimes applied,
where a small set of derived predictors
(linear combinations of features) become the predictor variables.
Using methods of partial distance correlation, one could evaluate the
subsets (as predictor sets) without reducing the multivariate observation to
a real number via linear combination. The pdCor coefficient can be computed
for multivariate predictors (and for multivariate response).
\end{Example}

\section{Summary}

Partial distance covariance and partial distance correlation coefficients are measures of dependence of two
random vectors $X$, $Y$, controlling for a third random vector $Z$, where $X$, $Y$, and $Z$ are in arbitrary, not
necessarily equal dimensions. These definitions are based on an unbiased distance covariance statistic, replacing
the original double centered distance matrices of the original formula with $\mathcal U$-centered distance matrices. The
$\mathcal U$-centered distance covariance is the inner product in the Hilbert space $\mathscr H_n$ of $\mathcal U$-centered
distance matrices for samples of size $n$, and it is unbiased for the squared population distance covariance.

Each element in the Hilbert $\mathscr H_n$ is shown to be the $\mathcal U$-centered distance matrix of some configuration
of $n$ points in a Euclidean space $\mathbb R^p$, $p \leq n-2$. The proof and the solution are obtained through
application of theory of classical MDS. This allows one to define dCor and dCov for dissimilarity matrices, and
provides a statistically consistent test of independence based on the inner product statistic. This methodology
also provides a similar test of the hypothesis of zero partial distance correlation based on the inner product.

All pdCor and pdCov methods have been implemented and simulation studies
illustrate that the tests control the
type 1 error rate at its nominal level.
Power performance was superior compared with power of partial
correlation and partial Mantel tests for association.
Methods have been implemented and illustrated for non-Euclidean
dissimilarity matrices.
A `pdCor forward selection' method was applied to select variables for
a linear model, with a significance test as the stopping rule.
More sophisticated selection methods will be investigated in future work.
With the flexibility to handle multivariate response and/or multivariate
predictor, pdCor offers a new method to extend the variable selection toolbox.
Software is available in the \emph{energy} \cite{energy}
package for R, and the \emph{pdcor} \cite{pdcor} package for R.

\appendix

\section{Proofs of statements}

\subsection{Proof of Proposition \ref{P:unbiased}\label{prfUnbiased}}

Proposition \ref{P:unbiased} asserts that $(\widetilde A \cdot \widetilde B)$ is an unbiased estimator of the
population coefficient $\mathcal V^2(X, Y)$. When the terms of $(\widetilde A \cdot \widetilde B)$ are expanded,
we have a linear combination of terms $a_{ij}b_{kl}$. The expected values of these terms differ according to the
number of subscripts that agree. Define
\begin{align*}
 \alpha &:= E[a_{kl}]=E[|X-X'|], \qquad \beta := E[b_{kl}]=E[|Y-Y'|], \qquad k \neq l,\\
 \delta &:= E[a_{kl}b_{kj}] = E[|X-X'||Y-Y''|], \qquad j, k, l \; \mathrm{ distinct},\\
 \gamma &:= E[a_{kl}b_{kl}]=E[|X-X'||Y-Y'|], \qquad k \neq l,
\end{align*}
where $(X,Y), (X',Y'),$ $(X'',Y'')$ are iid. Due to symmetry, the expected value of each term in the expanded
expression $(\widetilde A \cdot \widetilde B)$ is proportional to one of $\alpha \beta$, $\delta$, or $\gamma$,
so the expected value of $(\widetilde A \cdot \widetilde B)$ can be written as a linear combination of $\alpha
\beta$, $\delta$, and $\gamma$.

The population coefficient can be written as (see \cite[Theorem 7]{sr09a})
\begin{align*}
 \mathcal V^2(X, Y) &=
  E[|X-X'||Y\!-\!Y'|] + E[|X-X'|]\,E[|Y\!-\!Y'|] - 2E[|X-X'||Y\!-\!Y''|]
\\ &= \gamma + \alpha \beta - 2 \delta.
\end{align*}

Let us adopt the notation $\widetilde a_{k.}:= \frac{a_{k.}}{n-2}$, $\widetilde a_{.l}:=\frac{a_{.l}}{n-2}$, and
$\widetilde a_{..}:= \frac{a_{..}}{(n-1)(n-2)}$, where $a_{k.}=\sum_{l=1}^n a_{kl}$, $a_{.l}=\sum_{k=1}^n
a_{kl}$, and $a_{..}=\sum_{k,l=1}^n a_{kl}$.
 Similarly define $\widetilde b_{k.}, \widetilde b_{.l}$, and $\widetilde b_{..}$.
 Then
\begin{align*}
n(n-3) & (\widetilde A \cdot \widetilde B) =
\sum\limits_{k \neq l} \left\{
  \begin{array}{rrrr}
  \phantom{+}  a_{kl}b_{kl} & -a_{kl} \widetilde b_{k.} & -a_{kl} \widetilde b_{.l} & + a_{kl} \widetilde b_{..}
  \\
  - \widetilde a_{k.}b_{kl} & +  \widetilde a_{k.} \widetilde b_{k.} &  + \widetilde a_{k.} \widetilde b_{.l} & -
  \widetilde a_{k.} \widetilde b_{..} \\
  - \widetilde a_{.l}b_{kl} & +  \widetilde a_{.l} \widetilde b_{k.} & + \widetilde a_{.l} \widetilde b_{.l} & -
  \widetilde a_{.l} \widetilde b_{..} \\
  + \widetilde a_{..}b_{kl} &  - \widetilde a_{..} \widetilde b_{k.} & - \widetilde a_{..} \widetilde b_{.l} & +
  \widetilde a_{..} \widetilde b_{..}  \\
  \end{array}
  \right\}
\end{align*}
\begin{align*}
  \begin{array}{lllll}
=  \sum\limits_{k\neq l} a_{kl}b_{kl}  & -\sum\limits_{k}a_{k.} \widetilde b_{k.} & -\sum\limits_{l} a_{.\,l}
\widetilde b_{.\,l} & +  a_{..} \widetilde b_{..} \\
  -\sum\limits_{k} \widetilde a_{k.}b_{k.} & +(n-1) \sum\limits_{k}   \widetilde a_{k.} \widetilde b_{k.} &
  +\sum\limits_{k \neq l}   \widetilde a_{k.} \widetilde b_{.\,l}       & -(n-1)\sum\limits_{k}  \widetilde
  a_{k.} \widetilde b_{..} \\
  -\sum\limits_{l}  \widetilde a_{.\,l}b_{.\,l} & +\sum\limits_{k \neq l}   \widetilde a_{.\,l} \widetilde b_{k.}
  & +(n-1)\sum\limits_{l}  \widetilde a_{.\,l} \widetilde b_{.\,l} & -(n-1)\sum\limits_{l}  \widetilde a_{.\,l}
  \widetilde b_{..} \\
  + \widetilde a_{..} b_{..}  &  -(n-1)\sum\limits_{k}  \widetilde a_{..} \widetilde b_{k.} &
  -(n-1)\sum\limits_{l}  \widetilde a_{..} \widetilde b_{.\,l} & +n(n-1) \widetilde a_{..} \widetilde b_{..}.
  \,\\
  \end{array}
\end{align*}

Let $$T_1=\sum_{k \neq l} a_{kl}b_{kl}, \qquad T_2=a_{..}b_{..}, \qquad T_3=\sum_k a_{k.}b_{k.}. $$ Then
\begin{align*}
n(n-3)(\widetilde A \cdot \widetilde B) &=
\left\{
  \begin{array}{lllll}
  T_1 & -\frac{T_3}{n-2}  & - \frac{T_3}{n-2} & + \frac{T_2}{(n-1)(n-2)}\\
  - \frac{T_3}{n-2} &+ \frac{(n-1)T_3}{(n-2)^2} &+\frac{T_2-T_3}{(n-2)^2} &- \frac{T_2}{(n-2)^2} \\
  - \frac{T_3}{n-2} &+ \frac{T_2-T_3}{(n-2)^2} &+ \frac{(n-1)T_3}{(n-2)^2} &- \frac{T_2}{(n-2)^2} \\
  + \frac{T_2}{(n-1)(n-2)} &- \frac{T_2}{(n-2)^2} &- \frac{T_2}{(n-2)^2} &+ \frac{nT_2}{(n-1)(n-2)^2}\\
  \end{array}
 \right\}
 \\ &= T_1 - \frac{T_2}{(n-1)(n-2)^2} -\frac{2T_3}{n-2}.
\end{align*}

It is easy to see that $E[T_1]=n(n-1)\gamma$. By expanding the terms of $T_2$ and $T_3$, and combining terms that
have equal expected values, one can obtain $$
 E[T_2] = n(n-2)\{(n-2)(n-3)\alpha \beta + 2 \gamma + 4(n-2)\delta\},
$$ and $$
 E[T_3] = n(n-1)\{(n-2)\delta + \gamma\}.
$$ Then
\begin{align*}
E[(\widetilde A \cdot \widetilde B)] &=
\frac{1}{n(n-3)} E\left[ T_1 - \frac{T_2}{(n-1)(n-2)^2} -\frac{2T_3}{n-2}\right] \\
&= \frac{1}{n(n-3)} \left\{ \frac{n^3-5n^2+6n}{n-2} \gamma + n(n-3)\alpha \beta + (6n-2n^2)\delta \right \} \\ &=
\gamma + \alpha \beta - 2 \delta = \mathcal V^2(X, Y).
\end{align*}
\hfill \qed

\subsection{Proof of Proposition \ref{P:pdcor2}\label{prfRxy}}

Here the inner product is (\ref{VU}) and the `vectors' are $\mathcal U$-centered elements of the Hilbert space $\mathscr
H_n$. Equation (\ref{R.alt}) can be derived from (\ref{pdcor2}) using similar algebra with inner products as used
to obtain the representation
\begin{equation*}
 \frac{\la x_{z^{\perp}} , y_{z^{\perp}} \ra}
{\sqrt{ \la x_{z^{\perp}} ,x_{z^{\perp}} \ra
\la y_{z^{\perp}} ,y_{z^{\perp}} \ra }}
= \frac{r_{xy} - r_{xz}r_{yz}} {\sqrt{1-r_{xz}^2}\sqrt{1-r_{yz}^2}}.
\end{equation*}
for the linear correlation $r$ (see e.g.\ Huber \cite{huber81}). The details for $\pdCor$ are as follows. If
either $|\widetilde A|$ or $|\widetilde B|=0$ then $(\widetilde A,\widetilde B)=0$ so by definition
$R^*(x,y;z)=0$, $R^*_{xz}R^*_{yz}=0$, and (\ref{R.alt}) is also zero.

If $|\widetilde A||\widetilde B|\neq 0$ but $|\widetilde C|=0$, then $R^*_{x,z}=R^*_{y,z}=0$. In this case
$P_{z^\perp}(x)=\widetilde A$ and $P_{z^\perp}(y)=\widetilde B$, so that
\begin{align*}
R^*(x, y; z)&=
  \frac{(P_{z^\perp}(x) \cdot P_{z^\perp}(y))}{|P_{z^\perp}(x)||P_{z^\perp}(y)|}
 = \frac{(\widetilde A \cdot \widetilde B)}{|\widetilde A||\widetilde B|} = R^*_{xy},
\end{align*}
which equals expression (\ref{R.alt}) since $R^*_{x,z}=R^*_{y,z}=0$.

Suppose that none of $|\widetilde A|,|\widetilde B|,|\widetilde C|$ are zero. Then
\begin{align*}
R^*(x, y; z)&=
  (P_{z^\perp}(x) \cdot P_{z^\perp}(y))
 \\ &=
 \left(\left\{ \widetilde A-\frac{(\widetilde A \cdot \widetilde C)}{(\widetilde C \cdot \widetilde C)} \ \cdot
 \widetilde C\right\} \cdot  \left\{ \widetilde B - \frac{(\widetilde B \cdot \widetilde C)}{(\widetilde C \cdot
 \widetilde C)} \ \cdot  \widetilde C\right\} \right)
\\ &= (\widetilde A \cdot  \widetilde B) -  \frac{2(\widetilde A \cdot \widetilde C)(\widetilde B \cdot
\widetilde C)}{(\widetilde C \cdot \widetilde C)} + \frac{(\widetilde A \cdot \widetilde C)(\widetilde B \cdot
\widetilde C)(\widetilde C \cdot \widetilde C)}{(\widetilde C \cdot \widetilde C)^2}
\\ &= |\widetilde A||\widetilde B| \left\{ \frac{(\widetilde A \cdot \widetilde B)}{|\widetilde A||\widetilde B|}
- \frac{(\widetilde A \cdot \widetilde C)(\widetilde B \cdot \widetilde C)}{|\widetilde A||\widetilde
C||\widetilde B||\widetilde C|} \right\}
\\ &= |\widetilde A||\widetilde B| \left\{ R^*_{xy} - R^*_{xz}R^*_{yz} \right\}.
\end{align*}
Similarly, in the denominator of (\ref{pdcor2}) we have
\begin{align*}
& \sqrt{|\widetilde A|^2 \left( 1 - (R^*_{xz})^2 \right)
 |\widetilde B|^2 \left( 1 - (R^*_{yz})^2 \right) }
\\ &= |\widetilde A||\widetilde B| \sqrt{(1-(R^*_{xz})^2)(1-(R^*_{yz})^2)}.
\end{align*}
Thus, if the denominator of (\ref{pdcor2}) is not zero, the factor $|\widetilde A||\widetilde B|$ cancels from
the numerator and denominator and we obtain (\ref{R.alt}).
\hfill \qed

\subsection{Proof of Lemma \ref{lemmaU}\label{prfLemmaU}}
The sum of the first row of $\widetilde A$ is
\begin{align*}
 \widetilde A_{1.} &= \sum_{j=2}^n \left(a_{1j} - \frac{a_{1.}}{n-2} - \frac{a_{.j}}{n-2} +
\frac{a_{..}}{(n-1)(n-2)} \right)
\\&= a_{1.} - \frac{(n-1)a_{1.}}{n-2} - \frac{a_{..}-a_{.1}}{n-2} + \frac{(n-1)a_{..}}{(n-1)(n-2)}
=0.
\end{align*}
Similarly each of the rows of $\widetilde A$ sum to zero, and by symmetry the columns also sum to zero, which
proves statement (i).

Statements (ii) and (iii) follow immediately from (i). In (iv) let $b_{ij}$ be the $ij$-th element of $B$. Then
$b_{ii}=0$, and for $i \neq j$, $b_{ij}=\widetilde A_{ij} + c$. Hence $b_{i.}=b_{.j}=(n-1)c$, $i \neq j$ and
$b_{..}=n(n-1)c$. Therefore
\begin{align*}
 \widetilde B_{ij} = b_{ij}-\frac{2(n-1)c}{n-2} + \frac{n c}{n-2}
= \widetilde A_{ij} + c + \frac{(2-n)c}{n-2} = \widetilde A_{ij}, \qquad i \neq j,
\end{align*}
which proves (iv). \hfill \qed

\subsection{Proof of Lemma \ref{lemma3}\label{prfLemma3}}

It is clear that $c_1   A_X(x,x')$ is identical to
the double-centered dissimilarity $c_1 a(x,x')$.
It remains to show that the sum of two arbitrary elements
$  A_X(x,x') +   B_Y(y,y')$ is a double-centered
dissimilarity function. Let $T=[X,Y] \in \mathbb R^p \times \mathbb R^q$. Consider
the dissimilarity function $d(t,t'):= a(x,x')+  b(y,y')$,
where $t=[x,y]$ and $t'=[x',y']$.
Then
\begin{align*}
  D_X(t,t')&=
 a(x,x') + b(y,y')
 - \int [a(x,x')+b(y,y')]dF_T(t')
\\& \qquad - \int [a(x,x')+b(y,y')]dF_T(t)
\\& \qquad
+ \iint [a(x,x')+b(y,y')]dF_T(t')dF_T(t)
\\ &=
a(x,x') + b(y,y')
 -E[a(x,X')+b(y,Y')]
 \\&
\qquad - E[a(X,x')+b(Y,y')]
+ E[a(X,X')+b(Y,Y')]
\\&=   A_X(x,x') +   B_Y(y,y').
\end{align*}

\hfill \qed

\subsection{Proof of Theorem \ref{T:pdCor.pop}\label{prfThm3}}

\begin{proof} It is straightforward to check the first three special cases.

Case (i): If $Z$ is constant a.s. then $P_{Z^\perp}(X)=A_X$, and $P_{Z^\perp}(Y)=B_Y$.
In this case both $\mathcal R^*(X,Y;Z)$ and (\ref{pdCor.pop}) simplify to
$\mathcal R^2(X,Y)$.

Case (ii): If $X$ or $Y$ is constant a.s. and $Z$ is not a.s. constant,
then we have zero in both Definition \ref{defPDCOR} and (\ref{pdCor.pop}).

Case (iii): If none of the variables $X,Y,Z$ are a.s. constant, but
$|P_{Z^\perp}(X)|=0$ or $|P_{Z^\perp}(Y)|=0$, then $\mathcal R^*(X, Y;Z)=0$
by definition and $\mathcal R^2(X,Z)=1$ or
$\mathcal R^2(Y,Z)=1$. Thus (\ref{pdCor.pop}) is also zero by definition.

Case (iv): In this case none of the variables $X,Y,Z$ are a.s. constant,
and $|P_{Z^\perp}(X)||P_{Z^\perp}(Y)|>0$. Thus
\begin{align}\label{e:iv-1}
E[P_{Z^\perp}(X) & P_{Z^\perp}(Y)] = ( A_X-\alpha C_Z, B_Y-\beta C_Z )
 \\ \notag &= ( A_X,B_Y ) - \alpha ( B_Y,C_Z )
 - \beta ( A_X,C_Z ) + \alpha \beta ( C_Z, C_Z )
 \\ &= \notag
 \mathcal V^2(X, Y)
- \frac{\mathcal V^2(X,Z)\mathcal V^2(Y,Z)}{\mathcal V^2(Z,Z)}
\\& \notag \qquad
- \frac{\mathcal V^2(X,Z)\mathcal V^2(Y,Z)}{\mathcal V^2(Z,Z)}
 + \frac{\mathcal V^2(X,Z)\mathcal V^2(Y,Z) \mathcal V^2(Z,Z)}
 {\mathcal V^2(Z,Z) \mathcal V^2(Z,Z)}
  \\ \notag &=
 \mathcal V^2(X, Y)
- \frac{\mathcal V^2(X,Z)\mathcal V^2(Y,Z)}{\mathcal V^2(Z,Z)}.
\end{align}
Similarly
\begin{align} \label{e:iv-2}
|P_{Z^\perp}(X)|^2 &= E[P_{Z^\perp}(X) P_{Z^\perp}(X)] =
 \mathcal V^2(X, X)
- \frac{\mathcal V^2(X,Z)\mathcal V^2(X,Z)}{\mathcal V^2(Z,Z)}
\\ \notag &=  \mathcal V^2(X, X)
 - \alpha \mathcal V^2(X,Z)
 \\ \notag &= \mathcal V^2(X, X) \left( 1 - \frac{(\mathcal V^2(X,Z))^2}{\mathcal V^2(X, X) \mathcal V^2(Z, Z)} \right)
 \\ \notag &= \mathcal V^2(X, X) (1 - \mathcal R^4(X, Z)),
\end{align}
and
\begin{align} \label{e:iv-3}
|P_{Z^\perp}(Y)|^2 =
\mathcal V^2(Y, Y) (1 - \mathcal R^4(Y, Z)).
\end{align}
Hence, using (\ref{e:iv-1})--(\ref{e:iv-3})
\begin{align*}
\mathcal R^*(X,Y;Z)&= \frac{( P_{Z^\perp}(X) P_{Z^\perp}(Y) )}
{|P_{Z^\perp}(X)|\, |P_{Z^\perp}(Y)|}
\\ &=
\frac{ \mathcal V^2(X, Y)
- \frac{\mathcal V^2(X,Z)\mathcal V^2(Y,Z)}{\mathcal V^2(Z,Z)}}
{\sqrt{\mathcal V^2(X, X) (1 - \mathcal R^4(X, Z))} \sqrt{\mathcal V^2(Y, Y) (1 - \mathcal R^4(Y, Z))}}
\\&=
\frac{ \frac{\mathcal V^2(X, Y)}{\mathcal V(X,X) \mathcal V(Y,Y)}
- \frac{\mathcal V^2(X,Z)\mathcal V^2(Y,Z)}{\mathcal V(X,X) \mathcal V(Y,Y)\mathcal V^2(Z,Z)}}
{\sqrt{(1 - \mathcal R^4(X, Z))(1 - \mathcal R^4(Y, Z))}}
\\ &=
\frac{\mathcal R^2(X,Y) - \mathcal R^2(X,Z) \mathcal R^2(Y,Z)}
{\sqrt{(1 - \mathcal R^4(X,Z))(1 - \mathcal R^4(Y, Z))}}.
\end{align*}
Thus, in all cases Definition \ref{defPDCOR} and (\ref{pdCor.pop}) coincide.
\end{proof}

\section*{Acknowledgements}
The authors thank Russell Lyons for several helpful suggestions and
comments on a preliminary draft of this paper.


\begin{thebibliography}{00}

\bibitem{bss2004}
Baba, K., Shibata, R. and Sibuya, M. (2004). Partial correlation and
conditional correlation as measures of
conditional independence.
\emph{Australian and New Zealand Journal of Statistics} {\bf 46} (4):
657--664.

\bibitem{cailliez}
Cailliez, F. (1983). The analytical solution of the additive constant problem.
\emph{Psychometrika}, {\bf 48}, 343--349.

\bibitem{cc2001}
Cox, T. F. and Cox, M. A. A. (2001).
\emph{Multidimensional Scaling}, Second edition. Chapman and Hall.

\bibitem{degr2012}
Dueck, J., Edelmann, D. Gneiting, T. and Richards, D. (2012).
The affinely invariant distance correlation,
submitted for
publication. \url{http://arxiv.org/pdf/1210.2482.pdf}

\bibitem{feuer1993}
Feuerverger, A. (1993). A consistent test for bivariate dependence. \emph{
International Statistical Review}, {\bf 61}, 419--433.

\bibitem{ecodist}
  Goslee, S. C. and Urban, D. L. (2007). The ecodist package for
  dissimilarity-based analysis of ecological data. \emph{Journal of
  Statistical Software} {\bf 22}(7), 1--19.

\bibitem{gower1966}
Gower, J. C. (1966). Some distance properties of latent root and
vector methods used in multivariate analysis,
\emph{Biometrika} {\bf 53}, 325--328.

%\bibitem{he2013}
%Hastie, T. and Efron, B. (2013).
%lars: Least Angle Regression, Lasso and Forward
%Stagewise. R package version 1.2.
%\url{http://CRAN.R-project.org/package=lars}


\bibitem{htf2009}
Hastie, T., Tibshirani, R. and Friedman, J. (2009). \emph{Elements of
Statistical Learning, second edition},
Springer, New York.

\bibitem{huber81}
Huber, John (1981). Partial and Semipartial Correlation--A Vector Approach,
\emph{The Two-Year College Mathematics Journal}, {\bf 12}/2,
151--153, JSTOR: 3027381.

\bibitem{jh2013}
Josse, J. and Holmes, S. (2013). Measures of dependence between random
vectors and tests of independence.
Literature review.
\url{http://arxiv.org/abs/1307.7383}.

\bibitem{ppcor}
Kim, S. (2012). {ppcor}: Partial and Semi-partial (Part)
  correlation. R package version 1.0.
\url{http://CRAN.R-project.org/package=ppcor}

\bibitem{kkklw2012}
Kong, J., Klein, B. E. K., Klein, R., Lee, K., and Wahba, G. (2012).
Using distance correlation and {SS-ANOVA} to
assess associations of familial relationships, lifestyle factors,
diseases, and mortality,
\emph{Proc. of the National Acad. of Sciences},
 {\bf 109} (50), 20352--20357. DOI:10.1073/pnas.1217269109

\bibitem{legendre2000}
Legendre, P. (2000). Comparison of permutation methods for the
partial correlation and partial {M}antel tests,
\emph{J. Statist. Comput. Simul.}, {\bf 67}, 37--73.

\bibitem{ll2012}
Legendre, P. and Legendre, L. (2012).
\emph{Numerical Ecology}, 3rd English Edition, Elsevier.

\bibitem{lzz2012}
Li, R., Zhong, W. and Zhu, L. (2012). Feature Screening via
Distance Correlation Learning,
\emph{Journal of the American Statistical Asssociation},
{\bf 107}/499, 1129--1139. DOI: 10.1080/01621459.2012.695654

\bibitem{lyons2013}
Lyons, R. (2013). Distance covariance in metric spaces,
\emph{Ann. Probab.} {\bf 41}/5 , 3284--3305.

\bibitem{mantel1967}
Mantel, N. (1967). The detection of disease clustering and a
generalized regression approach, \emph{Cancer Res.},
{\bf 27}, 209--220.

\bibitem{mardia1978}
Mardia, K. V. (1978) Some properties of classical multidimensional scaling.
\emph{Communications in Statistics: Theory and Methods}, {\bf 7}(13),
1233--1241.

\bibitem{mkb1979}  %%% Chapter 14 Theorem 14.2.1 pp. 397--402.
Mardia, K. V., Kent, J. T. and Bibby, J. M. (1979).
\emph{Multivariate Analysis}, London: Academic Press.

%\bibitem{mrr2014a}
%Martinez-Gomez, E., Richards, M. T., and Richards, D. St. P. (2014).
%Distance correlation methods for discovering associations in large
%astrophysical databases, \emph{Astrophysical Journal}, {\bf 781}:39.

%\bibitem{mrr2014b}
%Martinez-Gomez, E., Richards, M. T., and Richards, D. St. P. (2014).
%Interpreting the distance correlation results for the COMBO-17
%survey, \emph{Astrophysical Journal Letters}, (in press).
%\url{http://arxiv.org/abs/1402.3230v2}

\bibitem{vegan}
  Jari Oksanen, F. Guillaume Blanchet, Roeland Kindt, Pierre Legendre,
  Peter R. Minchin, R. B. O'Hara, Gavin L. Simpson, Peter Solymos, M.
  Henry H. Stevens and Helene Wagner (2013). {vegan}: Community Ecology
  Package. R package version 2.0-7.
  \url{http://CRAN.R-project.org/package=vegan}

\bibitem{piepho}
Piepho, H. P. (2005). Permutation tests for the correlation
among genetic distances and measures of heterosis,
\emph{Theor. Appl. Genet.}, {\bf 111}, 95--99, DOI 10.1007/s00122-005-1995-7.

  \bibitem{R}
{R Core Team} (2013). R: A language and environment for statistical
  computing. R Foundation for Statistical Computing, Vienna, Austria.
  \url{http://www.R-project.org/}.

\bibitem{reif}
Reif, J. C., A. E. Melchinger, X. C. Xia, M. L. Warburton,
D. A. Hoisington, S. K. Vasal, G. Srinivasan, M. Bohn,
and M. Frisch (2003). Genetic distance based on simple
sequence repeats and heterosis in tropical maize
populations, \emph{Crop Science} {\bf 43}, no. 4, 1275--1282.

%\bibitem{riesz}
%Riesz, F. and Sz.-Nagy, B. (1990).
%\emph{Functional Analysis}, Dover Publications.

 \bibitem{energy}
Rizzo, M. L. and Sz\'ekely, G. J. (2014). {energy}: {E}-statistics
  (energy statistics). R package version 1.6.1.
  \url{http://CRAN.R-project.org/package=energy}

 \bibitem{pdcor}
Rizzo, M. L. and Sz\'ekely, G. J. (2013). {pdcor}: Partial
  distance correlation. R package version 1.0.0.

\bibitem{schoenberg}
Schoenberg, I. J. (1935). Remarks to {M}aurice {F}r\'echet's
article ``Sur la definition axiomatique d'une classe
d'espace distanci\'es vectoriellement applicable sur l'espace
de Hilbert.'' \emph{Ann. Math.},
 {\bf 36}, 724--732.

\bibitem{ssgf2013}
Sejdinovic, D., Sriperumbudur, B., Gretton, A., and Fukumizu, K. (2013).
Equivalence of distance-based and RKHS-based statistics in hypothesis
testing, \emph{Annals of Statistics}, {\bf 41}:5, 2263--2291.

\bibitem{sls1986}
Smouse, P. E., Long, J. C. and Sokal, R. R. (1986).
Multiple regression and correlation extensions of the Mantel
test of matrix correspondence.
\emph{Systematic Zoology}, {\bf 35}:62, 7--632.


\bibitem{prostatepaper}
Stamey, T. A., Kabalin, J. N., McNeal, J. E., Johnstone,
I. M., Freiha, F., Redwine, E. A. and Yang, N. (1989).
Prostate specific antigen in the diagnosis and treatment of
adenocarcinoma of the prostate:
II. radical prostatectomy treated patients,
\emph{Journal of Urology} {\bf 141}(5), 1076--1083.

\bibitem{srb07}
Sz\'ekely, G. J., Rizzo, M. L. and Bakirov, N. K. (2007) Measuring
and testing independence by correlation of
distances, \emph{Ann. Statist.}
{\bf 35}:6, 2769--2794.
DOI: 10.1214/009053607000000505

\bibitem{sr09a} Sz\'ekely, G. J. and Rizzo, M. L.  (2009).
 Brownian Distance Covariance, \emph{Annals of Applied Statistics},
 {\bf 3}(4), 1236--1265.
DOI: 10.1214/09-AOAS312

\bibitem{sr2012}
Sz\'ekely. G. J. and Rizzo, M. L. (2012). On the uniqueness of
distance covariance. \emph{Statistics \& Probability Letters},
{\bf 82}:12, 2278--2282.
DOI: 10.1016/j.spl.2012.08.007

\bibitem{sr2013a}
Sz\'ekely, G. J. and Rizzo, M. L. (2013). The distance
correlation t-test of independence in high dimension,
\emph{J. Multivar. Anal.}, {\bf 117}, 193--213.
DOI:  10.1016/j.jmva.2013.02.012

\bibitem{sr2013b}
Sz\'ekely, G. J. and Rizzo, M. L. (2013). Energy statistics:
statistics based on distances.
\emph{Journal of Statistical Planning and Inference},
{\bf 143}:8, 1249--1272.
DOI: 10.1016/j.jspi.2013.03.018

\bibitem{torg1958}
Torgerson, W. S. (1958). \emph{Theory and Methods of Scaling}, New York: Wiley.

%\bibitem{MASS}
%Venables, W. N. and Ripley, B. D. (2002).
%\emph{Modern Applied Statistics
%  with S, Fourth Edition}, Springer, New York.
%ISBN 0-387-95457-0.

\bibitem{wc13}
Wermuth, N. and Cox, D. R. (2013). Concepts and a case study for a
flexible class of graphical Markov models. In
Becker, C., Fried, R. and Kuhnt, S. (eds.) \emph{Robustness and
complex data structures. Festschrift in honour of
Ursula Gather.} Springer, Heidelberg, 331--350.

\bibitem{yh1938}
Young, G. and Householder, A. S. (1938). Discussion of a set of
points in terms of their mutual distances.
\emph{Psychometrika}, {\bf 3}, 19--22.



\end{thebibliography}
\end{document}